\documentclass[a4paper,UKenglish,cleveref, autoref,thm-restate]{lipics-v2019}

\usepackage{wrapfig}
\usepackage{pdfpages}

\newtheorem{assumption}{Assumption}

\newcounter{linecounter}
\newcommand{\linenumbering}{\ifthenelse{\value{linecounter}<10}{(0\arabic{linecounter})}{(\arabic{linecounter})}}
\renewcommand{\line}[1]{\refstepcounter{linecounter}\label{#1}\linenumbering}
\newcommand{\resetline}[1]{\setcounter{linecounter}{0}#1}
\renewcommand{\thelinecounter}{\ifnum \value{linecounter} >
9\else 0\fi \arabic{linecounter}}

\newcommand{\remove}[1]{}
\usepackage{xspace}
\newcommand {\CAS} {{\sf Compare\&Swap}\xspace}
\newcommand {\FAI} {{\sf Fetch\&Inc}\xspace}
\newcommand {\TAS} {{\sf Test\&Set}\xspace}
\newcommand {\SWAP} {{\sf Swap}\xspace}
\newcommand {\R} {{\sf Read}\xspace}
\newcommand {\W} {{\sf Write}\xspace}
\newcommand {\COUNT} {{\sf Counter}\xspace}
\newcommand {\INC} {{\sf Increment}\xspace}
\newcommand {\Ren} {{\sf Rename}\xspace}
\newcommand {\Pop} {{\sf Pop}\xspace}
\newcommand {\Push} {{\sf Push}\xspace}
\newcommand {\Enq} {{\sf Enqueue}\xspace}
\newcommand {\Deq} {{\sf Dequeue}\xspace}
\newcommand {\Lin} {{\sf Lin}\xspace}
\newcommand {\SetLin} {{\sf SetLin}\xspace}
\newcommand {\IntLin} {{\sf IntLin}\xspace}
\newcommand {\SeqStack} {{\sf Seq-Stack}\xspace}
\newcommand {\SetSeqStack} {{\sf Set-Conc-Stack}\xspace}
\newcommand {\SeqQueue} {{\sf Seq-Queue}\xspace}
\newcommand {\SetSeqQueue} {{\sf Set-Conc-Queue}\xspace}
\newcommand {\IntSeqQueue} {{\sf Int-Conc-Queue}\xspace}
\newcommand{\AC}[1]{{#1}}

\nolinenumbers

\title{Relaxed Queues and Stacks \\ from Read/Write Operations}
\titlerunning{Relaxed Queues and Stacks from Read/Write Operations}

\author{Armando Casta\~neda}{Instituto de Matem\'aticas, UNAM}{armando.castaneda@im.unam.mx}{}{Supported by UNAM-PAPIIT project IN108720.}

\author{Sergio Rajsbaum}{Instituto de Matem\'aticas, UNAM}{rajsbaum@matem.unam.mx}{}{Supported by UNAM-PAPIIT project IN106520.}

\author{Michel Raynal}{Institut Universitaire de France, IRISA-Universit\'e de Rennes and Polytechnic Univ. of Hong Kong}{michel.raynal@irisa.fr}{}{Supported by French ANR project DESCARTES  (16-CE40-0023-03).}

\authorrunning{A. Casta\~neda, S. Rajsbaum and M. Raynal}

\Copyright{Armando Casta\~neda, Sergio Rajsbaum and Michel Raynal} 

\hideLIPIcs



\begin{document}

\maketitle

\vspace{-0.2cm}
\begin{abstract}
Considering asynchronous shared memory systems in which any number of
processes may crash, this work identifies and formally defines
relaxations of queues and stacks that can be non-blocking or wait-free
while being implemented using only read/write operations.
Set-linearizability
and Interval-linearizability are used to specify the relaxations formally, and
  precisely identify the subset of executions which preserve the
  original sequential behavior.  The relaxations allow for an item to
  be returned  more than once by different operations, but only in
  case of concurrency; we call such a property \emph{multiplicity}.
  The stack implementation is wait-free, while the queue
  implementation is non-blocking.  
  Interval-linearizability is used to describe a queue with multiplicity, with
  the additional  relaxation that a dequeue operation can return
  \emph{weak-empty}, which means that the queue \emph{might} be empty.
  We present a read/write wait-free interval-linearizable algorithm of 
  a concurrent queue.  As far as we know, this work is the
  first that provides  formalizations of the notions
  of multiplicity and weak-emptiness, which can be implemented on top
  of read/write registers only.
\end{abstract}

\section{Introduction}

In the context of asynchronous crash-prone systems where processes
communicate by accessing a shared memory, linearizable implementations
of concurrent counters, queues, stacks, pools, and other concurrent
data structures~\cite{MSchapter07} need extensive synchronization
among processes, which in turn jeopardizes performance and scalability.
Moreover, it has been formally shown that this cost is sometimes
unavoidable, under various specific
assumptions~\cite{AGHK09,AGHKMV11,EHS12}.  However, often applications
do not require all guarantees offered by a linearizable sequential
specification~\cite{S11}.  Thus, much research has focused on
improving performance of concurrent data structures by relaxing their
semantics.  Furthermore, several works have focused on relaxations for
queues and stacks, achieving significant performance improvements
(e.g., ~\cite{HLHPSKS13,HKPSS13,KPRS12,S11}).

It is impossible however to implement queues and stacks with only
\R/\W operations, without relaxing their specification.  This is
because queues and stacks have consensus number two (i.e. they allow
consensus to be solved among two processes but not three), while the
consensus number of \R/\W operations is only one~\cite{H91}, hence too
weak to wait-free implement queues and stacks.
Thus, atomic {\sf Read-Modify-Write} operations, such as \CAS or \TAS,
are required in any queue or stack implementation.  To the best of our
knowledge, even relaxed versions of queues or stacks have not been
designed that avoid the use of {\sf Read-Modify-Write} operations.

In this article, we are interested in exploring if there are
meaningful relaxations of queues and stacks that can be implemented
using only simple \R/\W operations, namely, if there are non-trivial
relaxations with consensus number one.
Hence, this work is a theoretical investigation of the power of the
crash  \R/\W model for relaxed data structures.

\subsection{Contributions}
We identify and formally define relaxations of queues and stacks that can be
implemented using only \R/\W operations.
We consider queue and stack relaxations with \emph{multiplicity},
 where an item can be extracted by more than one dequeue {or pop}
 operation, instead of exactly once.  
 However,  this may happen only in the presence of concurrent
 operations.  As already argued~\cite{MVS09}, this type of relaxation
 {could} be useful in a wide range of applications, such as parallel
 garbage collection, fixed point computations in program analysis,
 constraint solvers (e.g. SAT solvers), state space search exploration
 in model checking, as well as integer and mixed programming solvers.

 \AC{One of the main challenges
   in designing relaxed data structures lies in  the
difficulty of formally specifying what is meant by ``relaxed specification''.}
To provide a formal specification of our relaxations, we use
\emph{set-linearizability}~\cite{N94} and
\emph{interval-linearizability}~\cite{CRR18}, specification methods
that are useful to specify the behavior of a data structure in
concurrent patterns of operation invocations, instead of only in
sequential patterns.  Using these specification methods, we are able
to precisely state in which executions the relaxed behavior of the
data structure should take place, and demand a strict behavior (not
relaxed), in other executions, especially when operation invocations
are sequential.

\paragraph*{First Contribution:}
We define a \emph{set-concurrent stack with multiplicity}, in which
no items are lost, all items are pushed/popped in LIFO order but an
item can be popped by multiple operations, which are then  concurrent.
We define a \emph{set-concurrent queue with multiplicity} similarly.  In both
cases we present set-linearizable implementations based only on
\R/\W operations.  The stack implementation is wait-free~\cite{H91}, while
the queue implementation is non-blocking~\cite{HW90}.

Our set-concurrent implementations imply \R/\W solutions for
\emph{idempotent work-stealing}~\cite{MVS09} 
and \emph{$k$-FIFO}~\cite{KPRS12} queues and stacks.

\paragraph*{Second Contribution:}

We define an interval-concurrent queue with
a \emph{weak-emptiness check}, which behaves like a classical
sequential queue with the exception that a dequeue operation can
return a control value denoted \emph{weak-empty}.  Intuitively, this
value means that the operation was concurrent with dequeue operations
that took the items that were in the queue when it started, thus the
queue might be empty.  First, we describe a wait-free
interval-linearizable implementation based on \FAI and \SWAP operations.
Then, using the techniques in our set-linearizable stack and queue
implementations, we obtain a wait-free interval-linearizable
implementation using only \R/\W operations.

Our interval-concurrent queue with weak-emptiness check is motivated
by a theoretical question that has been open for more than two
decades~\cite{AWW93}: it is unknown if there is a wait-free
linearizable queue implementation based on objects with consensus
number two (e.g.~\FAI or \SWAP), for any number of processes.  There
are only such non-blocking implementations in the literature, or
wait-free implementations for restricted cases
(e.g.~\cite{ACH18,E09,L01,M04,DBF05}).  Interestingly, our
interval-concurrent queue allows us to go from non-blocking to
wait-freedom.  Our interval-concurrent queue models the
\emph{tail-chasing} problem that one faces when trying to obtain a
wait-free queue implementation from objects with consensus number two.

Since we are interested in the computability power of \R/\W
  operations to implement relaxed concurrent objects (that otherwise
  are impossible), our algorithms are presented in an idealized
  shared-memory computational model.  
\AC{We hope these algorithms will help to develop a better understanding of
  fundamentals that can derive solutions for real multicore architectures,
  with good  performance and scalability.}

\subsection{Related Work}

It has been frequently pointed out that classic concurrent data
structures have to be relaxed in order to support scalability, and
examples are known showing how natural relaxations on the ordering
guarantees of queues or stacks can result in higher performance and
greater scalability~\cite{S11}.  Thus, for the past ten years there
has been a surge of interest in relaxed concurrent data structures
from practitioners (e.g.~\cite{NLP2013sosp}).  Also, theoreticians
have identified inherent limitations in achieving high scalability in
the implementation of linearizable
objects~\cite{AGHK09,AGHKMV11,EHS12}.

Some articles relax the sequential specification of traditional data
structures, while others relax their correctness condition
requirements.  As an example of relaxing the requirement of a
sequential data structure,~\cite{HKPSS13,KLP13,KPRS12,PRKS11} present
a \emph{$k$-FIFO} queue (called \emph{out-of-order} in~\cite{HKPSS13})
in which elements may be dequeued out of FIFO order up to a constant
$k\geq 0$.  A family of relaxed queues and stacks is introduced
in~\cite{ST16}, and studied from a computability point of view
(consensus numbers).  
It is defined in \cite{HKPSS13} the \emph{$k$-stuttering} relaxation
of a queue/stack, where an item can be returned by a dequeue/pop
operation without actually removing the item, up to $k \geq 0$ times,
even in sequential executions.  Our queue/stack with multiplicity is a
stronger version of $k$-stuttering, 
\AC{in the sense that an item can be returned by two operations
if and only if the operations are concurrent}.  Relaxed priority queues (in the
flavor of~\cite{ST16}) and associated performance experiments are
presented in~\cite{AKLS15,ZMS19}.
 
Other works design a weakening of the consistency condition. For instance,  
{\it quasi-linearizability}~\cite{AKY10}, which models
relaxed data structures through a distance function from valid
sequential executions. This work provides examples of
quasi-linearizable concurrent implementations that outperform state of
the art standard implementations. 
A \emph{quantitative} relaxation framework to formally specify relaxed
objects is introduced in~\cite{HLHPSKS13,HKPSS13} where relaxed
queues, stacks and priority queues are studied.  This framework is
more powerful than quasi-linearizability.
It is shown in~\cite{TW18} that linearizability and three data type
relaxations studied in~\cite{HKPSS13}, $k$-Out-of-Order, $k$-Lateness,
and $k$-Stuttering, can {also be} defined as consistency conditions.
The notion of {\it local linearizability} is introduced
in~\cite{HHHKLPSSV16}. It is a relaxed consistency condition
that is applicable to container-type concurrent data structures
like pools, queues, and stacks.
{
The notion of \emph{distributional linearizability}~\cite{AKLN18}
captures \emph{randomized} relaxations.
This formalism is applied to MultiQueues~\cite{RSD15}, a family of
concurrent data structures implementing relaxed concurrent
priority queues.
}

\AC{The previous works use relaxed specifications,
but still sequential, while we relax the specification to make it
concurrent (using set-linearizability and interval-linearizability).}

The notion of {\it idempotent work stealing} is
introduced in~\cite{MVS09}, where LIFO, FIFO and
double-ended set implementations are presented;
 these implementations exploit the relaxed semantics to
deliver better performance than usual work stealing algorithms.  
Similarly to our queues and stacks with multiplicity,
the \emph{idempotent} relaxation means that each
inserted item is eventually extracted at least once, instead of
exactly once. In contrast to our work, the algorithms presented
in~\cite{MVS09} use {\sf Compare\&Swap} (in the {\sf Steal} operation).
Being a practical-oriented work, formal specifications of the
implemented data structures are not given.

\subsection{Organization} 

The article is organized as follows.
Section~\ref{sec-preliminaries} presents the model of computations
and the
correctness conditions, namely,  linearizability, set-linearizability and
interval-linearizability.  Section~\ref{sec-set-stack} introduces the
notion of set-concurrent stack with multiplicity and presents a
read/wait wait-free solution of it, while Section~\ref{sec-set-queue}
defines the set-concurrent queue with multiplicity and shows a
non-blocking wait/free implementation.  Some consequences of the
set-concurrent queue and stack implementations are discussed in
Section~\ref{sec-implications}.  The new interval-concurrent queue
with weak-emptiness check and its implementation are presented in
Section~\ref{sec-int-seq-queue}.  Section~\ref{sec-final-discussion}
concludes the paper with a final discussion.

\section{Preliminaries}
\label{sec-preliminaries}

\subsection{Model of Computation}

We consider the standard concurrent system model with $n$
\emph{asynchronous} processes, $p_1, \hdots, p_n$, which may
\emph{crash} at any time during an execution, namely, a process that
crashes stops taking steps.  The \emph{index} of process $p_i$ is $i$.
Processes communicate with each other by invoking \emph{atomic}
operations on shared \emph{base objects}.  A base object can provide
atomic \R/\W operations (such an object is henceforth called a
\emph{register}), or more powerful atomic {\sf Read-Modify-Write}
operations, such as \FAI, \SWAP or \CAS.

The operation $R.\SWAP(x)$ atomically reads the current value of $R$,
sets its value to $x$ and returns $R$'s old value.  The operation
$R.\FAI()$ atomically adds~$1$ to the current value of $R$ and returns
the previous value. The operation $R.\CAS(new, old)$ is a conditional
replacement operation that atomically checks if the current value of
$R$ is equal to $old$, and if so, replaces it with $new$ and returns
{\sf true}; otherwise, $R$ remains unchanged and the operation returns
{\sf false}.

A \emph{(high-level) concurrent object}, or \emph{data type}, is,
roughly speaking, defined by a state machine consisting of a set of
states, a finite set of operations, and a set of transitions between
states. The specification does not necessarily have to be
\emph{sequential}, namely, (1) a state might have pending operations
and (2) state transitions might involve several invocations.  The
following subsections formalize this notion and the different types of
objects. 

An \emph{implementation} of a concurrent object $T$ is a distributed algorithm
$\mathcal A$ consisting of local state machines $A_1, \hdots, A_n$.  Local
machine $A_i$ specifies which operations on base objects $p_i$
executes in order to return a response when it invokes a high-level
operation of $T$.  Each of these base objects operation invocations is
a \emph{step}.

An \emph{execution} of $\mathcal A$ is a possibly infinite sequence of
steps, namely, executions of base objects operations, plus invocations
and responses to high-level operations of the concurrent object $T$,
with the following properties:

\begin{enumerate}

\item Each process
  is sequential. It first invokes a high-level operation, and only when
  it has a corresponding response, it can invoke another high-level
  operation, i.e., executions are \emph{well-formed}.

\item For any invocation to an operation {\sf op}, denoted $inv({\sf op})$, of a process $p_i$,
  the steps of $p_i$ between that invocation and its corresponding
  response (if there is one), denoted $res({\sf op})$, are steps that
  are specified by $\mathcal A$ when $p_i$ invokes $\sf op$.
\end{enumerate}

An operation in an execution is \emph{complete} if
both its invocation and response appear in the execution.
An operation is \emph{pending}
if only its invocation appears in the execution.
A process is \emph{correct} in an execution if it takes infinitely many steps.
For sake of simplicity, and without loss of generality, 
we identify the invocation of an operation with its first step, and 
its response with its last step.

In subsequent sections, we will formally define and implement relaxed 
versions the classical queues and stacks. For sake of simplicity, 
and without loss of generality, we will suppose that in every execution an
item can be enqueued/pushed at most once.

An implementation is \emph{wait-free} if every process completes each
 each operation it invokes~\cite{H91}.  An
implementation is \emph{non-blocking} if whenever processes take
steps and at least one of them does not crash,
at least one of the operations terminates~\cite{HW90}
(formally, in every infinite execution, infinitely many invocations are
completed).  Thus, a wait-free implementation is non-blocking but not
necessarily vice versa.

The \emph{consensus number} of a shared object $O$ is the maximum
number of processes that can solve the well-known \emph{consensus}
problem, using any number of instances of $O$ in addition to any
number of \R/\W registers~\cite{H91}.  Consensus numbers induce the
\emph{consensus hierarchy} where objects are classified according
their consensus numbers.  The simple \R/\W operations stand at the
bottom of the hierarchy, with consensus number one; these operations
are the least expensive ones in real multicore architectures.  At the
top of the hierarchy we find operations with infinite consensus
number, like \CAS, that provide the maximum possible coordination.

\subsection{The Linearizability Correctness Condition}

\emph{Linearizability}~\cite{HW90} is the standard notion used to define a
correct concurrent implementation of an object defined by a sequential 
specification (see below). 
Intuitively, an execution is linearizable if operations can be
ordered sequentially, without reordering non-overlapping operations,
so that their responses satisfy the specification of the implemented object
(see below).

A \emph{sequential specification} of a concurrent object $T$ is a
state machine specified through a transition function $\delta$. Given
a state $q $ and an invocation $inv({\sf op})$, $\delta(q, inv({\sf
  op}))$ returns the tuple $(q', res({\sf op}))$ (or a set of tuples
if the machine is \emph{non-deterministic}) indicating that the
machine moves to state $q'$ and the response to $\sf op$ is $res({\sf op}$).
\AC{In our specifications,  $res({\sf op})$ is written as a tuple $\langle {\sf op} : r \rangle$,
where $r$ is the output value of the operation.}
The sequences of invocation-response tuples, $\langle
inv({\sf op}): res({\sf op}) \rangle$, produced by the state machine
are its \emph{sequential executions}.

To formalize linearizability we define a partial order $<_\alpha$ on
the completed operations of an execution $\alpha$:
${\sf op} <_\alpha {\sf op}'$ if and only if $res({\sf op})$ precedes
$inv({\sf op}')$ in $\alpha$.
Two operations are \emph{concurrent},
denoted ${\sf op} ||_\alpha {\sf op}'$, if they are incomparable by $<_\alpha$.
The execution is \emph{sequential} if $<_\alpha$ is a total order.

\begin{definition}[Linearizability]
Let $\mathcal A$ be an implementation of a concurrent object $T$.
An execution $\alpha$ of $\mathcal A$ is \emph{linearizable} if there is
a sequential execution $S$ of $T$ such that
\begin{enumerate}

\item
$S$ contains every completed operation of $\alpha$ and might contain
  some pending operations.  Inputs and outputs of invocations and
  responses in $S$ agree with inputs and outputs in $\alpha$.

\item
For every two completed operations {\sf op} and ${\sf op}'$ in
$\alpha$, if ${\sf op} <_\alpha {\sf op}'$, then {\sf op} appears
before ${\sf op}'$ in $S$.
\end{enumerate}

We say that $\mathcal A$
is \emph{linearizable} if each of its executions is linearizable.
\end{definition}

\subsection{The Set-Linearizability Correctness Condition}

To formally specify our relaxed queues and stacks, we use the
formalism provided by the set-linearizability and
interval-linearizability consistency
conditions~\cite{CRR18,N94}.  
Roughly speaking,
set-linearizability allows us to linearize several operations in the same
point, namely, all these operations are executed concurrently, while
interval-linearizability allows operations to be linearized
concurrently with several non-concurrent
operations. Figure~\ref{fig-example-linear} schematizes the
differences between the three consistency conditions where each
double-end arrow represents an operation execution.
It is known that set-linearizability has strictly more expressiveness power than
linearizability, and interval-linearizability is strictly more
powerful than set-linearizability Moreover, as linearizability, both
set-linearizability and  interval-linearizability
are \emph{composable} (also called \emph{local})~\cite{CRR18}.

A \emph{set-concurrent specification} of a concurrent object differs
from a sequential execution in that $\delta$ receives as input the
current state $q$ of the machine and a set 
$Inv = \{ inv({\sf op_1}), \hdots, inv({\sf op}_t) \}$ of operation invocations,
\AC{and $\delta(q, Inv)$ returns $(q', Res)$, where $q'$ is
the next state and $Res = \{ res({\sf op_1}), \hdots, res({\sf op}_t)\}$ 
are the responses to the invocations in $Inv$.
Intuitively, all operations ${\sf op_1}, \hdots, {\sf op}_t$ are performed
concurrently and move the machine from state $q$ to $q'$.}
The sets $Inv$ and
$Res$ are called \emph{concurrency classes}.  
\AC{Observe that a set-concurrent specification in which 
all concurrency classes have a single element corresponds to a sequential specification.}

\begin{figure}[ht]
\begin{center}
\vspace{0.4cm}
\includegraphics[scale=0.7]{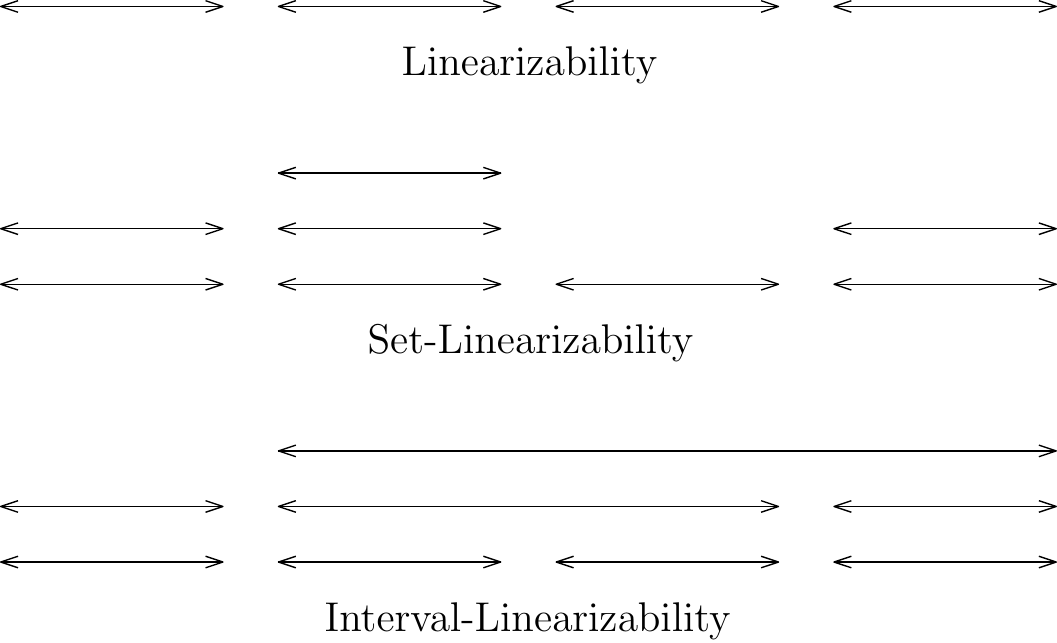}
\caption{\small Linearizability requires a total order
  on the operations, set-linearizability allows several operations to
  be linearized at the same linearization point, while 
  interval-linearizability allows an operation to be decomposed into
  several linearization points.
}
\label{fig-example-linear}
\end{center}
\end{figure}

\begin{definition}[Set-linearizability]
Let $\mathcal A$ be an implementation of a concurrent object $T$.
An execution $\alpha$ of $\mathcal A$ is \emph{set-linearizable} if there is
a set-concurrent execution $S$ of $T$ such that
\begin{enumerate}

\item
$S$ contains every completed operation of $\alpha$ and might contain
  some pending operations.  Inputs and outputs of invocations and
  responses in $S$ agree with inputs and outputs in $\alpha$.

\item
For every two completed operations {\sf op} and ${\sf op}'$ in
$\alpha$, if ${\sf op} <_\alpha {\sf op}'$, then {\sf op} appears
before ${\sf op}'$ in $S$.
\end{enumerate}

We say that $\mathcal A$
is \emph{set-linearizable} if each of its executions is set-linearizable.
\end{definition}

\subsection{The Interval-Linearizability Correctness Condition}

In an \emph{interval-concurrent specification}, some operations might be
pending in a given state~$q$, 
\AC{namely, the state records that there is an operation of a process without response.} 
We now have that in $(q', Res) = \delta(q, Inv)$,
some of the operations that are  pending in $q$ might still be pending in $q'$
and operations invoked in
$Inv$ may be  pending in $q'$, therefore $Res$ contains the responses to the
operations that are completed when moving from $q$ to $q'$.

\begin{definition}[Interval-linearizability]
Let $\mathcal A$ be an implementation of a concurrent object $T$.  An execution
$\alpha$ of $\mathcal A$ is \emph{interval-linearizable} if there is an
interval-concurrent execution $S$ of $T$ such that
\begin{enumerate}

\item
$S$ contains every completed operation of $\alpha$ and might contain
  some pending operations.  Inputs and outputs of invocations and
  responses in $S$ agree with inputs and outputs in $\alpha$.

\item
For every two completed operations {\sf op} and ${\sf op}'$ in
$\alpha$, if ${\sf op} <_\alpha {\sf op}'$, then {\sf op} appears
before ${\sf op}'$ in $S$.
\end{enumerate}

We say that $\mathcal A$ is \emph{interval-linearizable} if each of its
executions is interval-linearizable.
\end{definition}

\section{Set-Concurrent Stacks with Multiplicity}
\label{sec-set-stack}

By the \emph{universality} of consensus~\cite{H91}, we know that, for
every concurrent object there is a linearizable wait-free
implementation of it, for any number of processes, using \R/\W
registers and base objects with consensus number~$\infty$,
e.g. \CAS~\cite{HS08,R13,T06}.  However, the resulting implementation
might not be efficient because first, as it is universal, the
construction does not exploit the 
semantics of the particular object, and \CAS 
may be an expensive base operation. Moreover, such an approach would 
prevent us from investigating the power and the limit of the  \R/\W world, 
(as it was done for Snapshot object for which there are several
linearizable wait-free  \R/\W efficient implementations,
e.g.~\cite{AADGMS93, AACE15,AHR95,IR12}), 
and find  accordingly meaningfull  \R/\W-based  specifications of relaxed 
sequential specifications with efficient implmentations.

\subsection{A Wait-free Linearizable Stack from Consensus Number Two}
Afek, Gafni and Morisson
proposed in~\cite{AGM07} a simple linearizable wait-free
stack implementation for $ n \geq 2$ processes, using \FAI and \TAS
base objects, whose consensus number is 2.  Figure~\ref{figure:stack}
contains a slight variant of this algorithm that uses \SWAP and
\emph{readable} \FAI objects, both with consensus number
2.\footnote{The authors themselves explain  in~\cite{AGM07}
  how to replace \TAS with
  \SWAP.}  A \Push operation reserves a slot in $Item$ by atomically
reading and incrementing $Top$ (Line~\ref{M01}) and then places its
item in the corresponding position (Line~\ref{M02}). A \Pop operation
simply reads the $Top$ of the stack (Line~\ref{M04}) and scans down
$Items$ from that position (Line~\ref{M05}), trying to obtain an item
with the help of a \SWAP operation (Lines~\ref{M06} and~\ref{M07});
if the operation cannot get a item (a non-$\bot$ value), it returns empty
(Line~\ref{M09}). In what follows, we call this implementation \SeqStack.
It is worth {mentioning} that, although \SeqStack has a simple structure,
its linearizability proof is far from trivial,
the difficult part being  proving that items are taken in LIFO order.

In a formal sense, \SeqStack is the best we can do, from the
perspective of the consensus hierarchy: if there were a wait-free
(or non-blocking) linearizable implementation based only on \R/\W
registers, we could solve consensus among two processes in the
standard way, by popping a value from the stack initialized to a
single item containing a predefined value ${\tt winner}$; this is a
contradiction as consensus cannot be solved from \R/\W
registers~\cite{HS08,R13,T06}.  Therefore, there is no \emph{exact}
wait-free linearizable stack implementation from \R/\W registers only.
However, we could search for \emph{approximate} solutions.

Below, we show a formal definition of the notion of a \emph{relaxed}
set-concurrent stack and prove that it can be wait-free implemented
from \R/\W registers.  Informally, our solution consists in
implementing relaxed versions of \FAI and \SWAP with \R/\W registers, and
plug these implementations in \SeqStack.

\begin{figure}[ht]
\centering{ \fbox{
\begin{minipage}[t]{150mm}
\scriptsize
\renewcommand{\baselinestretch}{2.5} \resetline
\begin{tabbing}
aaaaa\=aaa\=aaa\=aaa\=aaa\=aaa\=aaa\=\kill 

{\bf Shared Variables:}\\

$~~$ $Top:$ {\sf Fetch\&Inc} base object initialized to 1\\

$~~$  $Items[1, \hdots ]:$ array of \SWAP base objects initialized to $\bot$\\ \\

{\bf Operation}  $\Push(x_i)$ {\bf is} \\

\line{M01} \> $top_i \leftarrow Top.\FAI()$\\

\line{M02} \> $Items[top_i].\W(x_i)$\\

\line{M03} \> {\bf {\sf return} } {\sf true}\\

{\bf end} \Push\\ \\

{\bf Operation}  $\Pop()$ {\bf is} \\

\line{M04} \> $top_i \leftarrow Top.\R() - 1$\\

\line{M05} \> {\bf for} $r_i \leftarrow top_i$ {\bf down to} 1 {\bf do}\\

\line{M06}  \>\> $x_i \leftarrow Items[r_i].\SWAP(\bot)$\\

\line{M07} \>\> {\bf if} $x_i \neq \bot$ {\bf then {\sf return}} $x_i$ {\bf end if} \\

\line{M08} \> {\bf end for}\\

\line{M09} \> {\bf {\sf return}} $\epsilon$\\

{\bf end} \Pop

\end{tabbing}
\end{minipage}
}
\caption{\small Stack implementation \SeqStack of Afek, Gafni and Morisson~\cite{AGM07} (code for process $p_i$).}
\label{figure:stack}
}
\end{figure}

\subsection{A Set-linearizable \R/\W Stack with Multiplicity}

Roughly speaking, our relaxed stack allows concurrent \Pop operations
to obtain the same item, but all items are returned in LIFO order, and
no pushed item is lost.  Formally, our set-concurrent stack is
specified as follows:

\begin{definition}[Set-Concurrent Stack with Multiplicity]
\label{def-set-stack}
The universe of items that can be pushed is $\mathbf{N} = \{ 1, 2,
\hdots \}$, and the set of states $Q$ is the infinite set of strings
$\mathbf{N}^*$.  The initial state is the empty string,
denoted~$\epsilon$.  In state $q$, the first element in $q$ represents
the top of the stack, which might be empty if $q$ is the empty string.
The transitions are the following:
\begin{enumerate}

\item For $q \in Q$, $\delta(q, \Push(x)) = (x \cdot q, \langle \Push(x) : {\sf true} \rangle)$.

\item For $q \in Q$, $1\leq t\leq n$ and $x \in \mathbf{N}:$
$\delta(x \cdot q, \{ \Pop_1(), \hdots, \Pop_t() \}) = (q, \{ \langle \Pop_1(): x \rangle, \hdots, \langle \Pop_t(): x \rangle \})$.

\item $\delta(\epsilon, \Pop() ) = (\epsilon, \langle \Pop(): \epsilon \rangle)$.
\end{enumerate}
\end{definition}

\begin{remark}
Every execution of the set-concurrent stack with all its concurrency
classes containing a single operation is an execution of the
sequential stack.
\end{remark}

The following lemma shows that any algorithm implementing the
set-concurrent stack keeps the behavior of a sequential stack in
several cases. In fact, the only reason the implementation does not
provide linearizability is due only to the \Pop operations that are
concurrent.

\begin{lemma}
\label{lemma-props-set-seq-stack}
Let $A$ be any set-linearizable implementation of the set-concurrent
stack with multiplicity.  Then,
\begin{enumerate}

\item All sequential executions of $A$ are executions of the sequential stack.

\item All executions with no concurrent \Pop operations are
  linearizable with respect to the sequential stack.

\item All executions with \Pop operations returning distinct values
  are linearizable with respect to the sequential stack.

\item If \Pop operations return the same value in an execution,
  then they are concurrent.
\end{enumerate}
\end{lemma}

\begin{proof}
Consider any execution $E$ of $A$ and a set-linearization $\SetLin(E)$
of it.  The definition of the set-concurrent stack implies that if all
concurrency classes in $\SetLin(E)$ have a single operation, then
$\SetLin(E)$ is an execution of the sequential stack.  We prove each
item separately:
\begin{enumerate}

\item Since $E$ is sequential, all concurrency classes in $\SetLin(E)$
  have a single operation.

\item If there are no concurrent \Pop operations in $E$, then every
  concurrency class of $\SetLin(E)$ contains at most one \Pop
  operation. By the specification of the set-sequential stack, every
  \Push operation appears alone in its concurrency class.  Thus, every
  concurrency class of $\SetLin(E)$ contains a single operation.

\item A similar reasoning implies that every concurrency class of
  $\SetLin(E)$ contains a single operation.

\item If any pair of \Pop return distinct values, then, the definition
  of the set-concurrent stack implies that every \Pop operation
  appears alone in its concurrency class. As observed before, the
  definition of the objects also implies that the same happens with
  \Push operations.  Thus, every concurrency class of $\SetLin(E)$
  contains a single operation.
\end{enumerate}

\end{proof}

\begin{figure}[ht]
\centering{ \fbox{
\begin{minipage}[t]{150mm}
\scriptsize
\renewcommand{\baselinestretch}{2.5} \resetline
\begin{tabbing}
aaaaa\=aaa\=aaa\=aaa\=aaa\=aaa\=aaa\=\kill 

{\bf Shared Variables:}\\

$~~$ $Top:$ \R/\W wait-free linearizable \COUNT base object initialized to 1\\

$~~$ $Items[1, \hdots ][1, \hdots, n]:$ array of \R/\W registers initialized to $\bot$\\ \\

{\bf Operation}  $\Push(x)$ {\bf is} \\

\line{L01} \> $top_i \leftarrow Top.\R()$\\

\line{L02} \> $Top.\INC()$\\

\line{L04} \> $Items[top_i, i].\W(x)$\\

\line{L05} \> {\bf {\sf return} } {\sf true}\\

{\bf end} \Push\\ \\

{\bf Operation}  $\Pop()$ {\bf is}  \\

\line{L06} \> $top_i \leftarrow Top.\R() - 1$\\

\line{L07} \> {\bf for} $r_i \leftarrow top_i$ {\bf down to} 1 {\bf do}\\

\line{L08} \> \> {\bf for} $s_i \leftarrow n$ {\bf down to} 1 {\bf do}\\

\line{L09}  \> \> \> $x_i \leftarrow Items[r_i][s_i].\R()$\\

\line{L11}  \> \> \> {\bf if} $x_i \neq \bot$ {\bf then}\\

\line{L10}  \> \> \> \> $Items[r_i][s_i].\W(\bot)$\\

\line{L11'}  \> \> \> \> {\bf {\sf return}} $x_i$\\

\line{L11''}  \> \> \> {\bf end if}\\

\line{L12} \>  \>{\bf end for}\\

\line{L13} \> {\bf end for}\\

\line{L14} \> {\bf {\sf return}} $\epsilon$\\

{\bf end} \Pop

\end{tabbing}
\end{minipage}
}
  \caption{\small  \R/\W wait-free set-concurrent stack
                   \SetSeqStack with multiplicity (code for process $p_i$).}
\label{figure:set-seq-stack}
}
\end{figure}

The algorithm in Figure~\ref{figure:set-seq-stack} is a
set-linearizable \R/\W wait-free implementation of the stack with
multiplicity, which we call \SetSeqStack.  This implementation is a
modification of \SeqStack. The \FAI operation in Line~\ref{M01} in
\SeqStack is replaced by a \R and \INC operations of a \R/\W wait-free
linearizable \COUNT, in Lines~\ref{L01} and~\ref{L02} in \SetSeqStack.
This causes a problem as two \Push operations can set the same value
in their $top_i$ local variables.  This problem is resolved with the
help of a two-dimensional array $Items$ in Line~\ref{L04}, which
guarantees that no pushed item is lost:
\AC{each row of $Items$ now has $n$ entries, each of them
  associated with one ond only one process.}
Similarly, the \SWAP
operation in Line~\ref{M06} in \SeqStack is replaced by \R and \W
operations in Lines~\ref{L09} and~\ref{L10} in \SetSeqStack, together
with the test in Line~\ref{L11} which ensures that a \Pop operation
modifies an entry in $Items$ only if an item has been written in it.
Thus, it is now possible that two distinct \Pop operations get the
same non-$\bot$ value, which is fine because this can only happen if
the operations are concurrent.  Object $Top$ in \SetSeqStack can be
any of the known \R/\W wait-free linearizable \COUNT
implementations\footnote{To the best of our knowledge, the best
  implementation is in~\cite{AAC12} with polylogarithmic step
  complexity, on the number of processes, provided that the number of
  increments is polynomial.}.

\begin{theorem}
\label{theo-set-seq-stack}
The algorithm \SetSeqStack (Figure~{\em\ref{figure:set-seq-stack}})
is a \R/\W wait-free set-linearizable
implementation of the stack with multiplicity.
\end{theorem}

\begin{proof} 
Since all base objects are wait-free, it follows directly from the
code that the implementation is wait-free.  Thus, we focus on proving
that the implementation is set-linearizable.  Let $E$ be any execution
of \SetSeqStack. Since the algorithm is wait-free, there is an
extension of $E$ in which all its pending operations are completed,
and no new operation is started. Any set-linearization of such
extension is a set-linearization of $E$. Thus, without loss of
generality, we can assume all operations in $E$ are completed.

The rest of the proof is a ``reduction'' that
proceeds as follows. First, we modify $E$ and
remove some of its operations to obtain another execution $G$ of the
algorithm.  Then, from $G$, we obtain an execution $H$ of \SeqStack,
and show that we can obtain a set-linearization $\SetLin(G)$ of $G$
from any linearization $\Lin(H)$ of $H$.  Finally, we add to
$\SetLin(G)$ the operations of $E$ that were removed to obtain a
set-linearization $\SetLin(E)$ of $E$.

We start with the following simple remarks:

\begin{remark}
Every \Push operation gets a unique pair $(top_i, i)$, hence no pushed
value is lost.
\end{remark}

\begin{remark}
If two \Push operations store their values at entries in the same row
of $Items$ (both set the same value in $top_i$ in Line~{\em\ref{L02}}),
they are concurrent.
\end{remark}

\begin{remark}
If two \Pop operations return the same value, they are concurrent.
\end{remark}

To obtain the execution $G$ mentioned above, we first obtain
intermediate executions $F$ and then $F'$, from which we derive $G$.

For any value $y \neq \epsilon$ that is returned by more than one \Pop
operation in $E$, we remove from $E$ all these operations
(invocations, responses and steps) except for the first one that
executes Line~\ref{L10}, i.e., the first among these operations that
marks $y$ as taken in $Items$.  Let $F$ be the resulting sequence. We
claim that $F$ is an execution of the algorithm: (1) any \Pop
operation reads $Top$ in Line~\ref{L06} and $Items[r_i][t_i]$ in
Line~\ref{L09}, thus no step of any other operation depends on such a
step of a removed \Pop operation, and (2) $F$ keeps the \Pop operation
that marks first $y$ as taken in Line~\ref{L10}, hence the subsequent
\W steps of the removed \Pop operations are superfluous.  As we only
removed some operations from $E$, the remaining operations in $F$
respect the partial order $<_E$, namely, $<_F \subseteq <_E$.
Since there are no two \Pop
operations in $F$ popping the same item $y \neq \epsilon$, then for
every \Pop operation we can safely move {backward}
each of its steps in
Line~\ref{L10} next to its previous step in Line~\ref{L09} (which
corresponds to the same iteration of the for loop in Line~\ref{L08}).
Thus, for every \Pop operation, Lines~\ref{L09} to~\ref{L10}
correspond to a \SWAP operation.  Let $F'$ denote the resulting
equivalent execution.

We now permute the order of some steps in $F'$ to obtain an execution
$G$ with the same operations and $<_G = <_{F'}$.
For each integer $b \geq 0$, let $t(b) \in [0, \hdots, n]$ be the
number of \Push operations in $F'$ that store their items in row
$Items[b]$. Namely, each of these operations obtains $b$ in its \R steps
in Line~\ref{L01}.  Let $\Push^b_1, \hdots, \Push^b_{t(b)}$ denote all
these operations.  For each $\Push^b_j$, let $x^b_j$ denote the item
the operation pushes, let $e^b_j$ denote its \R step in
Line~\ref{L01}, and let $ind^b_j$ be the index of the process that
performs operation $\Push^b_j$.  Hence, $\Push^b_j$ stores its item
$x^b_j$ in $Items[b][ind^b_j]$ when performs Line~\ref{L04}.  Without
loss of generality, let us suppose that $ind^b_1 < ind^b_2 < \hdots <
ind^b_{t(b)}$.  Observe the following:
\begin{itemize}

\item $\Push^b_1, \hdots, \Push^b_{t(b)}$ are concurrent.

\item By linearizability of $Top$, all $\R$ operations in
  Lines~\ref{L02} and~\ref{L06} read monotonically increasing
  values.  Therefore, for $a < b$, the step $e^a_k$ of any $\Push^a_k$
  appears before the step $e^b_j$ of any $\Push^b_j$.

\item Since all $e^b_1, \hdots, e^b_{t(b)}$ read the same value from
  $Top$, there are no $Top.\INC()$ steps in the shortest sub-string of
  $F'$ containing $e^b_1, \hdots, e^b_{t(b)}$.

\item Let $f^b$ be the step among $e^b_1, \hdots, e^b_{t(b)}$ that
  appears first in $F'$.  Then, the step in Line~\ref{L04} of any
  $\Push^b_j$ appears after $f^b$ in $F'$.
\end{itemize}

\begin{figure}[ht]
\begin{center}
\includegraphics[scale=0.65]{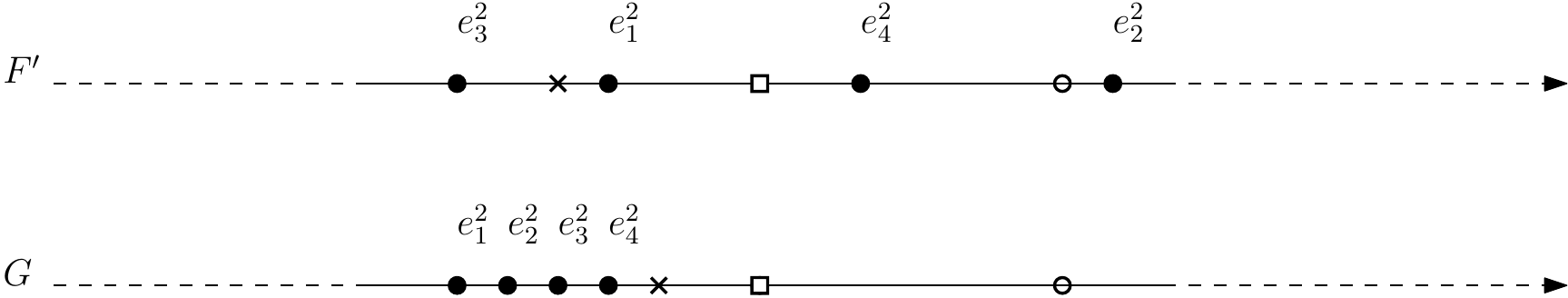}
\caption{\small Graphical description of the transformation from $F'$
  to $G$. In the example, the relative order of other steps
  (represented with the cross, box and empty circle) keep the same
  relative order.}
\label{fig-example-1}
\end{center}
\end{figure}

The last two items imply that moving forward each $e^b_j$ right after
$f^b$ produces another execution equivalent to $F'$.  Thus, we obtain
$G$ by moving forward all steps $e^b_1, \hdots, e^b_{t(b)}$ up to the
position of $f^b$, and place them in that order, $e^b_1, \hdots,
e^b_{t(b)}$, for every $b \geq 0$.  Figure~\ref{fig-example-1} shows a
graphical description of the transformation.  Observe that $<_G = <_{F'}$.

The main observation now is that $G$ already corresponds to an
execution of \SeqStack, if we consider the entries in $Items$ in the
their usual order (first row, then column).  We say that $Items[r][s]$
is \emph{touched} in $G$ if there is a \Push operation that writes its
item in that entry; otherwise, $Items[r][s]$ is \emph{untouched}.
Now, for every $b \geq 0$, in $G$ all $\Push^b_1, \hdots,
\Push^b_{t(b)}$ execute Line~\ref{L01} one right after the other, in
order $e^b_1, \hdots, e^b_{t(b)}$. Also, the items they push appear in
row $Items[b]$ from left to right in order $\Push^b_1, \hdots,
\Push^b_{t(b)}$.  Thus, we can think of the touched entries in row
$Items[b]$ as a column with the left most element at the bottom, and
pile all rows of $Items$ with $Items[0]$ at the
bottom. Figure~\ref{fig-example-2} depicts an example of the
transformation.  In this way, each $e^b_j$ corresponds to a \FAI
operation and every \Pop operations scans the touched entries of
$Items$ in the order \SeqStack does (note that it does not matter if
the operation start scanning in a row of $Items$ with no touched
entries, since untouched entries are immaterial).  Following this
idea, we do the following to obtain an execution $H$ of \SeqStack from
$G$:

\begin{figure}[ht]
\begin{center}
\includegraphics[scale=0.7]{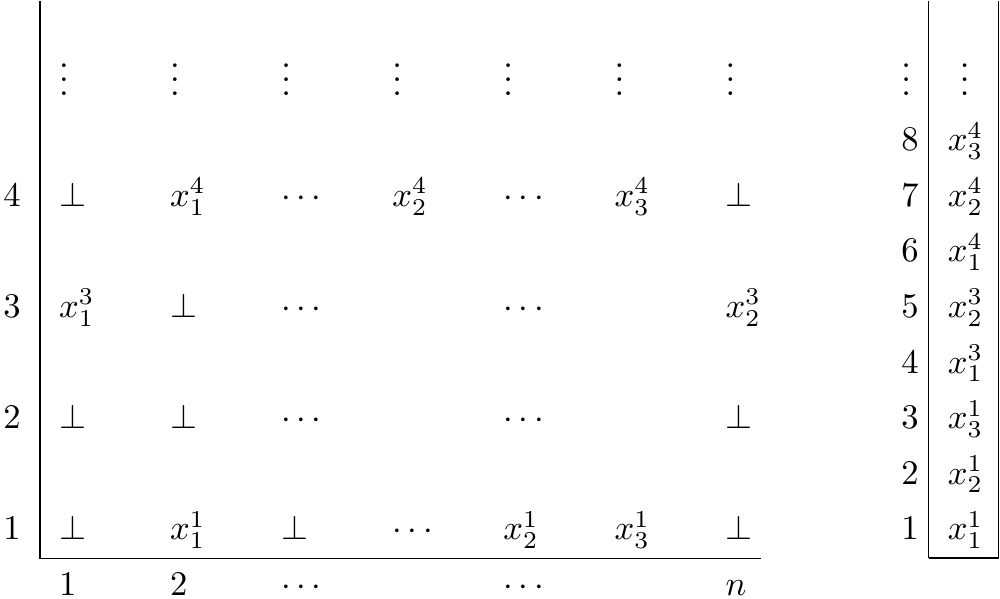}
\caption{\small An example of the codification of the one-dimensional
  array $Items$ of \SeqStack in the two-dimensional array $Items$ in
  \SetSeqStack.  The untouched entries are represented with $\bot$.}
\label{fig-example-2}
\end{center}
\end{figure}

\begin{itemize}
\item For each $\Push^b_j$:
\begin{enumerate}

\item Replace its step in Line~\ref{L01} (denoted $e^b_j$ above) with
  $top_i \leftarrow Top.\FAI()$, corresponding to Line~\ref{M01} of
  \SeqStack (Figure~\ref{figure:stack}).

\item Remove its step in Line~\ref{L02}.

\item Replace its step in Line~\ref{L04} with $Items[top_i].\W(x_i)$,
corresponding to Line~\ref{M02} of \SeqStack (Figure~\ref{figure:stack}).
\end{enumerate}

\item For each $\Pop$ operation:
\begin{enumerate}

\item Note that the step in Line~\ref{L06} directly corresponds to
  Line~\ref{M04} of \SeqStack (Figure~\ref{figure:stack}).  Thus, this
  step remains unchanged.
\item Let $e$ be any of its steps corresponding to Line~\ref{L09}. If
  $e$ reads a touched entry, then replace $e$ with $x_i \leftarrow
  Items[r_i].\SWAP(\bot)$, corresponding to Line~\ref{M06} of
  \SeqStack (Figure~\ref{figure:stack}); otherwise, remove $e$.

\item Remove all its steps corresponding to Line~\ref{L10}. 
\end{enumerate}
\end{itemize}

As already argued, $H$ is an execution of \SeqStack; furthermore, $G$
and $H$ have the same operations and $<_H = <_G$. 
Let $\Lin(H)$ be any linearization of \SeqStack. To conclude the proof of
the theorem, we obtain a set-linearization of $E$ from $\Lin(H)$. 
We have that $H$, $G$ and $F$ have the same operations and $<_H = <_G = <_F$,  
and then $\Lin(H)$ is indeed a set-linearization of $F$ and $G$
with each concurrency class having a single operation.  To obtain a
set-linearization $\SetLin(E)$ of $E$, we put every \Pop operation of
$E$ that is removed to obtain $F$, in the concurrency class of
$\Lin(H)$ with the \Pop operation that returns the same item.
The resulting set-concurrent execution, $\SetLin(E)$, respects $<_E$ because 
any two operations returning the same item are concurrent, as observed at the
beginning of the proof.  Therefore, $\SetLin(E)$ is a
set-linearization of $E$, and hence \SetSeqStack is set-linearizable.
\end{proof}

It is worth observing that indeed it is simple to prove that
\SetSeqStack is an implementation of the \emph{set-concurrent pool
  with multiplicity}, namely, Definition~\ref{def-set-stack} without
LIFO order (i.e. $q$ is a set instead of a string).  
The hard part in the previous proof is the LIFO order,
\AC{which is shown through a reduction to the (nontrivial) linearizability proof of 
\SeqStack~\cite{AGM07}.}

\AC{
\section{A Renaming-based Performance-related Improvement (when contention is small)}
}

When the contention on the shared memory accesses is small, 
a \Pop operation in \SetSeqStack might perform several ``useless'' \R
 operations in Line~\ref{L09}, as it scans all
 entries of $Items$ in every row while trying to get a non-$\bot$
 value, and some of these entries might never store an item in the
 execution (called untouched in the proof of
 Theorem~\ref{theo-set-seq-stack}).  The algorithm in
 Figure~\ref{figure:set-seq-stack-improve} mitigates this issue with the
 help of an array $Ren$ with instances of any \R/\W $f(n)$-adaptive
 renaming.  In \emph{$f(n)$-adaptive renaming}~\cite{ABDPR90},
 each process starts with its index as input
 and obtains a unique name in the space $\{1, \hdots, f(p)\}$, where
 $p$ denotes the number of processes participating in the execution.
 Several adaptive renaming algorithms have been proposed (see
 e.g.~\cite{CRR11}); a good candidate is the simple $(p^2/2)$-adaptive
 renaming algorithm of Moir and Anderson with $O(p)$ individual step
 complexity~\cite{MA95}.

\begin{figure}[ht]
\centering{ \fbox{
\begin{minipage}[t]{150mm}
\scriptsize
\renewcommand{\baselinestretch}{2.5} \resetline
\begin{tabbing}
aaaaa\=aaa\=aaa\=aaa\=aaa\=aaa\=aaa\=\kill 

{\bf Shared Variables:}\\

\> $Top:$ \R/\W wait-free linearizable \COUNT base object initialized to 1\\

\> $\mathit{NOPS}[1, \hdots]:$
array of \R/\W wait-free linearizable \COUNT base objects initialized to 0\\

\> $Ren[1, \hdots ]:$ array of instances of \R/\W $f(n)$-adaptive renaming\\

\> $Items[1, \hdots ][1, \hdots ]:$
           array of \R/\W registers initialized to $\bot$\\ \\

           {\bf Operation}  $\Push(x)$ {\bf is} \\
           
\line{S01} \> $top_i \leftarrow Top.\R()$\\
 
 \line{S02} \> $tiebreaker_i \leftarrow Ren[top_i].\Ren(i)$\\
 
 \line{S03} \> $Top.\INC()$\\
 
 \line{S04} \> $\mathit{NOPS}[top_i].\INC()$\\
 
 \line{S05} \> $Items[top_i, tiebreaker_i].\W(x)$\\
 
 \line{S06} \> {\bf {\sf return} } {\sf true}\\
 
 {\bf end} \Push\\ \\

 {\bf Operation}  $\Pop()$ {\bf is} \\
 
 \line{S07} \> $top_i \leftarrow Top.\R() - 1$\\
 
 \line{S09} \> {\bf for} $r_i \leftarrow top_i$ {\bf down to} $1$ {\bf do}\\
 
 \line{S08} \>  \> $nops_i \leftarrow \mathit{NOPS}[r_i].\R()$\\
 
 \line{S08'} \>  \> $max_i \leftarrow f(nops_i)$\\
 
 \line{S10} \> \> {\bf for} $s_i \leftarrow max_i$ {\bf down to} 1 {\bf do}\\
 
 \line{S11}  \> \> \> $x_i \leftarrow Items[r_i][s_i].\R()$\\
 
 \line{S13}  \> \> \> {\bf if} $x_i \neq \bot$ {\bf then}\\
 
 \line{S12}  \> \> \> \> $Items[r_i][s_i].\W(\bot)$\\
 
 \line{S13'}  \> \> \> \> {\bf {\sf return}} $x_i$\\
 
 \line{S13''}  \> \> \> {\bf end if}\\
 
 \line{S14} \>  \>{\bf end for}\\
 
 \line{S15} \> {\bf end for}\\
 
 \line{S16} \> {\bf {\sf return}} $\epsilon$\\
 
 {\bf end} \Pop
 
 \end{tabbing}
 \end{minipage}
 }
 \caption{\small An improved \R/\W wait-free set-concurrent stack with
   multiplicity (code for process $p_i$).}
 \label{figure:set-seq-stack-improve}
 }
 \end{figure}
 
 \Push operations storing their items in the same row $Items[b]$, 
 which has now infinite length,
 dynamically decide where in the row they store their items, with the
 help of $Ren[b].\Ren(\cdot)$ in Line~\ref{S02}.  Additionally, these
 operations announce in $\mathit{NOPS}[b]$ the number of processes
 that store values in row $Items[b]$, in Line~\ref{S04}, hence helping
 \Pop operations to scan only the segment of $Items[b]$ where there
 might be items, in {\bf for} loop in Line~\ref{S10}.
 The correctness proof of this
 implementation is very similar to the correctness proof of
 \SetSeqStack.  
 
 Note that if the contention is small, say $O(\log^x n)$, 
 every \Pop operation scans only the first entries $O(\log^{2x}
 n)$ of row $Items[b]$ as the processes storing items in that row
 rename in the space $\{1, \hdots, (\log^{2x} n)/2\}$, 
 using the Moir and Anderson $(p^2/2)$-adaptive renaming algorithm.
Finally, observe that $n$ does not to be known in the modified algorithm
(as in \SeqStack).

\section{Set-Concurrent Queues with Multiplicity}
\label{sec-set-queue}

\subsection{A Non-Blocking Linearizable Queue from Consensus Number Two}
\label{LF-lin-queue-CN=2}
We now consider the linearizable queue implementation in
Figure~\ref{figure:queue}, which uses objects with consensus number
two.  The idea of the implementation, which we call \SeqQueue, is
similar to that of \SeqStack in the previous section.  Differently
from \SeqStack, whose operations are wait-free, \SeqQueue has a
wait-free \Enq and a non-blocking \Deq.

\SeqQueue is a slight modification of the non-blocking queue
implementation of Li~\cite{L01}, which in turn is a variation of the
blocking queue implementation of Herlihy and Wing~\cite{HW90}.  Each
\Enq operation simply reserves a slot for its item by performing \FAI
to the tail of the queue, Line~\ref{P01}, and then stores it in
$Items$, Line~\ref{P02}.  A \Deq operation repeatedly tries to obtain
an item scanning $Items$ from position 1 to the tail of the queue
(from its perspective), Line~\ref{P07}; every time it sees an item has
been stored in an entry of $Items$, Lines~\ref{P09} and~\ref{P10}, it
tries to obtain the item by atomically replacing it with $\top$, which
signals that the item stored in that entry has been taken,
Line~\ref{P11}.  While scanning, the operation records the number of
items that has been taken (from its perspective), Line~\ref{P13}, and
if this number is equal to the number of items that were taken in the
previous scan, it declares the queue is empty, Line~\ref{P16}.  For
completeness, the correctness proof of \SeqQueue is in
Appendix~\ref{app-proof-seq-queue}.
Despite its simplicity, \SeqQueue's  linearizability proof is far from trivial.

\begin{figure}[ht]
\centering{ \fbox{
\begin{minipage}[t]{150mm}
\scriptsize
\renewcommand{\baselinestretch}{2.5} \resetline
\begin{tabbing}
aaaaa\=aaa\=aaa\=aaa\=aaa\=aaa\=aaa\=\kill 

{\bf Shared Variables:}\\

\> $Tail:$ {\sf Fetch\&Inc} base object initialized to 1\\

\> $Items[1, \hdots ]:$ array of \SWAP base objects initialized to $\bot$\\ \\

{\bf Operation}  $\Enq(x_i)$ {\bf is} \\

\line{P01} \> $tail_i \leftarrow Tail.\FAI()$\\

\line{P02} \> $Items[tail_i].\W(x_i)$\\

\line{P03} \> {\bf {\sf return} } {\sf true}\\

{\bf end} \Enq\\ \\

{\bf Operation}  $\Deq()$ {\bf is}\\

\line{P04} \> $taken'_i \leftarrow 0$\\

\line{P05} \> {\bf while} {\sf true} {\bf do}\\

\line{P06} \> \> $taken_i \leftarrow 0$\\

\line{P07} \> \> $tail_i \leftarrow Tail.\R()-1$\\

\line{P08} \> \> {\bf for} $r_i \leftarrow 1$ {\bf up to} $tail_i$ {\bf do}\\

\line{P09}  \> \> \> $x_i \leftarrow Items[r_i].\R()$\\

\line{P10}  \> \> \> {\bf if} $x_i \neq \bot$ {\bf then}\\

\line{P11}  \> \> \> \> $x_i \leftarrow Items[r_i].\SWAP(\top)$\\

\line{P12}  \> \> \> \>
            {\bf if} $x_i \neq \top$ {\bf then {\sf return}} $x_i$ {\bf end if}\\

\line{P13}  \> \> \> \> $taken_i \leftarrow taken_i + 1$\\

\line{P14}  \> \> \> {\bf end if}\\

\line{P15} \> \> {\bf end for}\\

\line{P16} \> \> {\bf if} $taken_i = taken'_i$ {\bf then} {\bf {\sf return}} $\epsilon$\\

\line{P17} \> \> $taken'_i \leftarrow taken_i$\\

\line{P18} \> {\bf end while}\\

{\bf end} \Deq

\end{tabbing}
\end{minipage}
}
\caption{\small Non-blocking linearizable queue 
  \SeqQueue from base objects with consensus number 2 (code for $p_i$).}
\label{figure:queue}
}
\end{figure}

Similarly to the case of the stack, \SeqQueue is optimal from the
perspective of the consensus hierarchy as there is no non-blocking
linearizable queue implementation from \R/\W operations only.  However, as
we will show below, we can obtain a \R/\W non-blocking implementation of a
set-concurrent queue with multiplicity.

\subsection{A Set-linearizable \R/\W Queue with Multiplicity}

Our relaxed queue follows a similar idea of that of the set-concurrent
stack in Definition~\ref{def-set-stack}: concurrent \Deq operations
might obtain the same item, but all items are returned in FIFO order,
and no enqueued item is lost.

\begin{definition}[Set-Concurrent Queue with Multiplicity]
\label{def-set-queue}
The universe of items that can be enqueued is $\mathbf{N} = \{ 1, 2,
\hdots \}$, and the set of states $Q$ is the infinite set of strings
$\mathbf{N}^*$.  The initial state is the empty string,
denoted~$\epsilon$.  In state $q$, the first element in $q$ represents
the head of the queue, which might be empty if $q$ is the empty
string.  The transitions are the following:
\begin{enumerate}

\item
  For $q \in Q$, $\delta(q, \Enq(x)) = (q \cdot x, \langle \Enq(x): {\sf true} \rangle)$.

\item For $q \in Q$, $1\leq t\leq n$, $x \in \mathbf{N}:$ 
  $\delta(x \cdot q, \{
  \Deq_1(), \hdots, \Deq_t() \}) = (q, \{ \langle \Deq_1(): x \rangle,
  \hdots,  \langle \Deq_t(): x \rangle \})$.

\item
  $\delta(\epsilon, \Deq() ) = (\epsilon, \langle \Deq(): \epsilon \rangle)$.
\end{enumerate}
\end{definition}

\begin{remark}
Every execution of the set-concurrent queue with all its concurrency
classes containing a single operation is an execution of the
sequential queue.
\end{remark}

The proof of the following lemma is similar to the proof of Lemma~\ref{lemma-props-set-seq-stack}.

\begin{lemma}
\label{lemma-props-set-seq-queue}
Let $A$ be any set-linearizable implementation of the set-concurrent
queue with multiplicity.  Then,
\begin{enumerate}

\item All sequential executions of $A$ are executions of the sequential queue.

\item All executions with no concurrent \Deq operations are
  linearizable with respect to the sequential queue.

\item All executions with \Deq operations returning distinct values
  are linearizable with respect to the sequential queue.

\item If two \Deq operations return the same value in an execution,
  then they are concurrent.
\end{enumerate}
\end{lemma}

The algorithm in Figure~\ref{figure:set-seq-queue} is a
set-linearizable \R/\W non-blocking implementation of a queue with
multiplicity, which we call \SetSeqQueue.  As for the case of the
stack before, we obtain \SetSeqQueue from \SeqQueue by: (1) replacing
the \FAI object in \SeqQueue with a \R/\W wait-free \COUNT, (2)
extending $Items$ to a matrix to handle collisions, and (3) simulating
the \SWAP operation with a \R followed by a \W.  The correctness proof
of \SetSeqQueue is similar to the correctness proof of \SetSeqStack in
Theorem~\ref{theo-set-seq-stack}.

\begin{figure}[ht]
\centering{ \fbox{
\begin{minipage}[t]{150mm}
\scriptsize
\renewcommand{\baselinestretch}{2.5} \resetline
\begin{tabbing}
aaaaa\=aaa\=aaa\=aaa\=aaa\=aaa\=aaa\=\kill 

{\bf Shared Variables:}\\

\> $Tail:$ \R/\W wait-free linearizable \COUNT base object initialized to 1\\

\> $Items[1, \hdots ][1, \hdots, n]:$
                        array of \R/\W registers initialized to $\bot$\\ \\

{\bf Operation}  $\Enq(x)$ {\bf is}\\

\line{T01} \> $tail_i \leftarrow Tail.\R()$\\

\line{T02} \> $Tail.\INC()$\\

\line{T04} \> $Items[tail_i, i].\W(x)$\\

\line{T05} \> {\bf {\sf return} } {\sf true}\\

{\bf end} \Enq\\ \\

{\bf Operation}  $\Deq()$ {\bf is} \\

\line{T06} \> $taken'_i \leftarrow 0$\\

\line{T07} \> {\bf while} {\sf true} {\bf do}\\

\line{T08} \> \> $taken_i \leftarrow 0$\\

\line{T09} \> \> $tail_i \leftarrow Tail.\R()-1$\\

\line{T10} \> \> {\bf for} $r_i \leftarrow 1$ {\bf up to} $tail_i$ {\bf do}\\

\line{T10'} \> \> \> {\bf for} $s_i \leftarrow 1$ {\bf up to} $n$ {\bf do}\\

\line{T11}  \> \> \> \> $x_i \leftarrow Items[r_i][s_i].\R()$\\

\line{T12}  \> \> \> \> {\bf if} $x_i \neq \bot$ {\bf then}\\

\line{T13}  \> \> \> \> \> $Items[r_i][s_i].\W(\top)$\\

\line{T14}  \> \> \> \> \>
       {\bf if} $x_i \neq \top$ {\bf then {\sf return}} $x_i$ {\bf end if}\\

\line{T15}  \> \> \> \> \> $taken_i \leftarrow taken_i + 1$\\

\line{T16}  \> \> \> \> {\bf end if}\\

\line{T17} \> \> \> {\bf end for}\\

\line{T17'} \> \> {\bf end for}\\

\line{T18} \> \> {\bf if} $taken_i = taken'_i$ {\bf then} {\bf {\sf return}} $\epsilon$ {\bf end if}\\\

\line{T19} \> \> $taken'_i \leftarrow taken_i$\\

\line{T20} \> {\bf end while}\\

{\bf end} \Deq

\end{tabbing}
\end{minipage}
}
  \caption{\small \R/\W non-blocking set-concurrent queue \SetSeqQueue
    with multiplicity (code for $p_i$).}
\label{figure:set-seq-queue}
}
\end{figure}

\begin{theorem}
\label{theo-set-seq-queue}
The algorithm \SetSeqQueue (Figure~{\em{\ref{figure:set-seq-queue}}})
is a \R/\W non-blocking set-linearizable
implementation of the queue with multiplicity.
\end{theorem}

\begin{proof}
First, observe that the \Enq method is wait-free. To prove that \Deq
is lock-free, it is enough to observe that the only way a \Deq
operations never terminates is because it sets larger values in
$taken_i$ at the end of each iteration of the while loop, which can
only happen if there are new \Enq operations and \Deq operations,
implying that infinitely many operations are completed.

The set-linearizability proof of the implementation is nearly the same
as in the proof of Theorem~\ref{theo-set-seq-stack}. Given any
execution $E$ without pending operations, we obtain in the same way
$F'$ and then $F'$ to obtain an ``equivalent'' execution $G$ of
\SetSeqQueue. Again, $G$ naturally corresponds to an execution of $H$
of \SeqStack, and hence we consider any linearization $\Lin(H)$ of
$H$, which is a set-linearization of $G$ in which all concurrency
classes have a single element. Finally, from $\Lin(H)$ we obtain a
set-linearization $\SetLin(E)$ of $E$ by adding the \Deq operations
that where removed from $E$.
\end{proof}

In fact, proving \SetSeqQueue implements the set-concurrent pool with multiplicity is simple,
the difficulty  comes from the FIFO order requirement of the queue, which is shown through a
simulation argument.


\section{Implications}
\label{sec-implications}

\subsection{Avoiding Costly Synchronization Operations/Patterns}

It is worth {observing} that \SetSeqStack and \SetSeqQueue allow us to
circumvent the {linearization-related} 
impossibility results in~\cite{AGHKMV11}, where it is
shown that every linearizable implementation of a queue or a stack, as
well as other {concurrent operation executions as encountered for example in}
 work-stealing, must
use either expensive {\sf Read-Modify-Write} operations (e.g. \FAI and \CAS)
or {\sf Read-After-Write patterns}~\cite{AGHKMV11} (i.e. a process writing
in a shared variable and then reading another shared variable, maybe
performing operation on other variables in between).

In the simplest \R/\W\,\COUNT implementation we are aware of, the
object is represented via a shared array $M$ with an entry per
process; process $p_i$ performs \INC by incrementing its entry,
$M[i]$, and \R by reading, one by one, the entries of $M$ and
returning the sum.  Using this simple \COUNT implementation, we obtain
from \SetSeqStack a set-concurrent stack implementation with
multiplicity, devoided of (1) {\sf Read-Modify-Write} operations, as
only \R/\W operations are used, and (2) {\sf Read-After-Write} patterns, as
in both operations, \Push and \Pop, a process first reads and then
writes.  It similarly happens with \SetSeqQueue.

\subsection{Work-stealing with multiplicity}

Our implementations also provide relaxed \emph{work-stealing}
solutions without expensive synchronization operation/patterns.
Work-stealing is a popular technique to implement load balancing in a
distributed manner, in which each process maintains its own
\emph{pool} of tasks and occasionally \emph{steals} tasks from the
pool of another process.  In more detail, a process can {\sf Put} and
{\sf Take} tasks in its own pool and {\sf Steal} tasks from another
pool.  To improve performance, \cite{MVS09}  introduced  the
notion of \emph{idempotent work-stealing} which allows a task to be
taken/stollen at least once instead of exactly once as in previous work.
Using this relaxed notion, three different solutions are presented in
that paper where the {\sf Put} and {\sf Take} operations avoid {\sf
  Read-Modify-Write} operations and {\sf Read-After-Write} patterns;
however, the {\sf Steal} operation still uses costly \CAS operations.

Our set-concurrent queue and stack implementations provide idempotent
work-stealing solutions in which no operation uses {\sf
  Read-Modify-Write} operations and {\sf Read-After-Write}
patterns. Moreover, in our solutions both {\sf Take} and {\sf Steal}
are implemented by \Pop (or \Deq), hence any process can invoke those
operations, allowing more concurrency. If we insist that {\sf Take}
and {\sf Steal} can be invoked only by the owner, $Items$ can be a
1-dimensional array.  Additionally, differently from~\cite{MVS09},
whose approach is practical, our queues and stacks with multiplicity
are formally defined, with a clear and simple semantics.

\subsection{Out-of-order queues and stacks with multiplicity}
The notion of a \emph{$k$-FIFO queue} is introduced in~\cite{KPRS12} 
(called $k$-out-of-order queue in~\cite{HKPSS13}),
in which items can be dequeued out of FIFO order up to an integer~$k
\geq 0$.  More precisely, dequeueing the oldest item may require up to
$k+1$ dequeue operations, which may return elements not younger than
the $k + 1$ oldest elements in the queue, or nothing even if the queue
is not empty.  \cite{KPRS12} also presented  a simple way to
implement a $k$-FIFO queue, through $p$ independent FIFO queue
linearizable implementations.  When a process wants to perform an
operation, it first uses a \emph{load balancer} to pick one of
the $p$ queues and then performs its operation.  The value of $k$
depends on $p$ and the load balancer.  Examples of load balancers are
round-robin load balancing, which requires the use of {\sf
  Read-Modify-Write} operations, and randomized load balancing, which
does not require coordination but can be computational locally
expensive.  As explained in~\cite{KPRS12}, the notion of a
\emph{$k$-FIFO stack} can be defined and implemented similarly.

We can relax the $k$-FIFO queues and stacks to include multiplicity,
namely, an item can be taken by several concurrent operations.  Using
$p$ instances of our set-concurrent stack or queue \R/\W
implementations, we can easily obtain set-concurrent implementations
of $k$-FIFO queues and stacks with multiplicity, where the use of
{\sf Read-Modify-Write} operations or {\sf Read-After-Write} patterns are in the load balancer.

\section{Interval-Concurrent Queues with Weak-Emptiness Check}
\label{sec-int-seq-queue}

A natural question is if in Section~\ref{sec-set-queue} we could start
with a wait-free linearizable queue implementation instead of
\SeqQueue, which is only non-blocking, and hence derive a wait-free
set-linearizable queue implementation with multiplicity.  It turns out
that it is an open question if there is a wait-free linearizable queue
implementation from objects with consensus number two. 
(Concretely, such an algorithm would show that the
  queue belongs to the {\sf Common2} family of
  operations~\cite{AWW93}.)  This question has been open for more than
two decades~\cite{AWW93} and there have been several papers proposing
wait-free implementations of restricted
queues~\cite{ACH18,E09,L01,M04,DBF05}, e.g., limiting the number of
processes that can perform a type of operations.

\begin{figure}[ht]
\centering{ \fbox{
\begin{minipage}[t]{150mm}
\scriptsize
\renewcommand{\baselinestretch}{2.5} \resetline
\begin{tabbing}
aaaaa\=aaa\=aaa\=aaa\=aaa\=aaa\=aaa\=\kill 

{\bf Shared Variables:}\\

\> $Tail:$ {\sf Fetch\&Inc} base object initialized to 1\\

\> $Items[1, \hdots ]:$ array of \SWAP base objects initialized to $\bot$\\ \\

{\bf Operation}  $\Enq(x_i)$ {\bf is} \\

\line{MM01} \> $tail_i \leftarrow Tail.\FAI()$\\

\line{MM02} \> $Items[tail_i].\W(x_i)$\\

\line{MM03} \> {\bf {\sf return} } {\sf true}\\

{\bf end} \Enq\\ \\

{\bf Operation}  $\Deq()$ {\bf is} \\

\line{MM04} \> $tail_i \leftarrow Tail.\R() - 1$\\

\line{MM05} \> {\bf for} $r_i \leftarrow 1$ {\bf up to} $tail_i$ {\bf do}\\

\line{MM06}  \> \> $x_i \leftarrow Items[r_i].\SWAP(\bot)$\\

\line{MM07}  \> \> {\bf if} $x_i \neq \bot$ {\bf then {\sf return}} $x_i$ {\bf end if}\\

\line{MM08} \> {\bf end for}\\

\line{MM09} \> {\bf {\sf return}} $\epsilon$\\

{\bf end} \Deq

\end{tabbing}
\end{minipage}
}
  \caption{\small A non-linearizable queue implementation
                               (code for process $p_i$).}
\label{figure:queue-incorrect}
}
\end{figure}

\subsection{The Tail-Chasing Problem}

One of the main difficulties to solve when trying to design such an
implementations using objects with consensus number two 
is that of reading the current position of the tail. 
This problem, which we call as \emph{tail-chasing}, can be easily
exemplified with the help of the \emph{non-linearizable} queue
implementation in Figure~\ref{figure:queue-incorrect}. The
implementation is similar to \SeqStack with the difference that \Deq
operations scan $Items$ in the opposite order, i.e. from the head to
the tail.

The problem with this implementation is that once a \Deq
has scanned unsuccessfully $Item$ (i.e., the items that were in the
queue were taken by ``faster'' operations), it returns $\epsilon$;
however, while the operation was scanning, more items could have been
enqueued, and indeed it is not safe to return $\epsilon$ as the queue
might not be empty.  Figure~\ref{fig-tail-chasing} describes an
execution of the implementation that cannot be linearized because there
is no moment in time during the execution of the \Deq operation
returning $\epsilon$ in which the queue is empty.  Certainly, this
problem can be solved as in \SeqQueue: read the tail and scan again;
thus, in order to complete, a \Deq operation is forced to \emph{chase}
the current position of the tail until it is sure there are no new items.

\begin{figure}[h]
\begin{center}
\includegraphics[width=12cm]{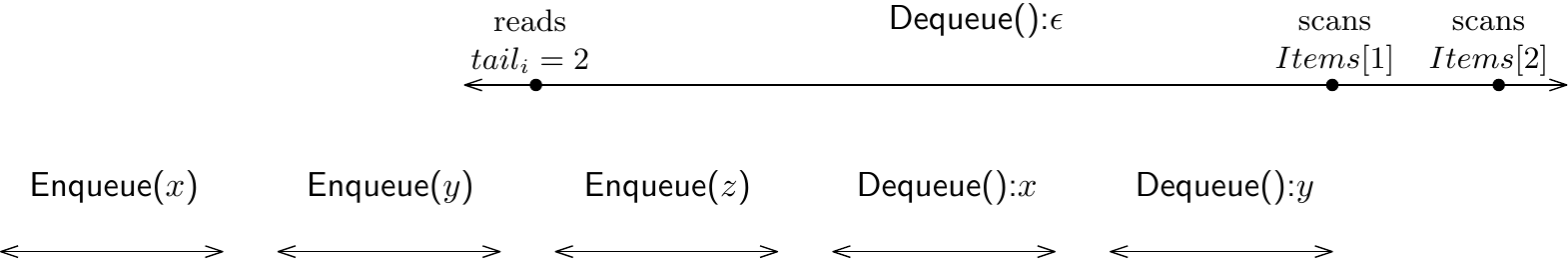}
\caption{\small An example of the tail-chasing problem.}
\label{fig-tail-chasing}
\end{center}
\end{figure}

Inspired by this problem, below we introduce a relaxed
interval-concurrent queue that allows a \Deq operation to return a
weak-empty value, with the meaning that the operation was not able
take any of the items that were in the queue when it started but it
was concurrent with all the \Deq operation that took those items,
i.e., it has a sort of \emph{certificate} that the items were taken,
and the queue might be empty. Then, we show that such a relaxed queue
can be wait-free implemented from objects with consensus number two.

\subsection{A Wait-Free Interval-Concurrent Queue with Weak-Emptiness}

Roughly speaking, in our relaxed interval-concurrent queue, the state
is a tuple $(q, P)$, where $q$ denotes the state of the queue and $P$
denotes the \emph{pending} \Deq operations that eventually return
\emph{weak-empty}, denoted $\epsilon^{\sf w}$.  More precisely, $P[i]
\neq \bot$ means that process $p_i$ has a pending \Deq operation.
$P[i]$ is a prefix of $q$ and represents the remaining items that have
to be dequeued so that the current \Deq operation of $p_i$ can return
$\epsilon^{\sf w}$.  
\Deq operations taking items from the queue, also remove the items from $P[i]$, 
and the operation of $p_i$ can return $\epsilon^{\sf w}$
only if $P[i]$ is $\epsilon$.  
Intuitively, the semantics of $\epsilon^{\sf w}$ is that the queue \emph{could} be
empty as all items that were in the queue when the operations started
have been taken. So this \Deq operation virtually occurs after all the
items have been dequeued.

\begin{figure}[ht]
\begin{center}
\includegraphics[width=10cm]{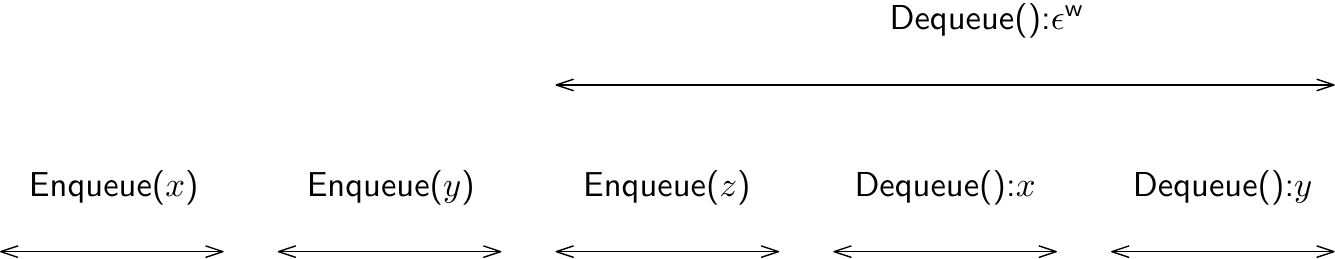}
\caption{\small An interval-concurrent execution with a \Deq
  operations returning weak-empty.}
\label{fig-exec-weak-empty}
\end{center}
\end{figure}

Figure~\ref{fig-exec-weak-empty} shows an example of an
interval-concurrent execution of our relaxed queue where the \Deq
operation returning $\epsilon^{\sf w}$ is allowed to return only when
$x$ and $y$ have been dequeued, as the queue contains those values
when the operations starts.  Observe this execution is an
interval-linearization of the execution obtained from
Figure~\ref{fig-tail-chasing} by replacing $\epsilon$ with
$\epsilon^{\sf w}$.

\begin{definition}[Interval-Concurrent Queue with Weak-Empty]
\label{def-weak-empty}
The universe of items that can be enqueued is $\mathbf{N} = \{ 1, 2,
\hdots \}$ and the set of states is $Q = \mathbf{N}^* \times
(\mathbf{N}^* \cup \{\bot\})^n$, with the initial state being
$(\epsilon, \bot, \hdots, \bot)$.  Below, a subscript denotes the ID
of the process invoking an operation.  The transitions are the
following:
\begin{enumerate}
\item
  For $(q, P) \in Q$, $0\leq t, \ell \leq n-1$,
  $\delta(q, P, \Enq(x), \Deq_{i(1)}(),
  \hdots, \Deq_{i(t)}())$ contains the transition $(q \cdot x, S,
  \langle \Enq(x): {\sf true} \rangle, \langle \Deq_{j(1)}():
  \epsilon^{\sf w} \rangle, \hdots, \langle \Deq_{j(\ell)}():
  \epsilon^{\sf w} \rangle)$, satisfying that
\begin{enumerate}

\item $i(k) \neq i(k')$, $i(k) \neq$ the id of the process invoking
$\Enq(x)$, and $j(k) \neq j(k')$,

\item for each $i(k)$, $P[i(k)] = \bot$,

\item
  for each $j(k)$, either $P[j(k)] = \epsilon$, or $P[j(k)] = \bot$
     and $q = \epsilon$ and $j(k) = i(k')$ for some~$k'$,
     
 \item for each $1 \leq s \leq n$, if there is a $k$ with $s = j(k)$, then $S[s] = \bot$;
 otherwise, if there is $k'$ with $s = i(k')$, $S[s] = q$, else $S[s] = P[s]$.

\end{enumerate}

\item For $(x \cdot q, P) \in Q$, $0 \leq t, \ell \leq n-1$, 
  $\delta(x \cdot q, P, \Deq(), \Deq_{i(1)}(), \hdots, \Deq_{i(t)}())$
  contains the transition  $(q, S, \langle \Deq(): x \rangle,
  \langle \Deq_{j(1)}(): \epsilon^{\sf w} \rangle, \hdots, \langle
  \Deq_{j(\ell)}(): \epsilon^{\sf w} \rangle)$, satisfying that
\begin{enumerate}

\item $i(k) \neq i(k')$, $i(k) \neq$ the id of the process invoking
  $\Deq()$, and $j(k) \neq j(k')$,

\item for each $i(k)$, $P[i(k)] = \bot$,

\item for each $j(k)$, either $P[j(k)] = x$, or $P[j(k)] = \bot$
  and $q = \epsilon$ and $j(k) = i(k')$ for some~$k'$,

\item for each $1 \leq s \leq n$, if there is a $k$ with $s = j(k)$, then $S[s] = \bot$;
otherwise, if there is $k'$ with $s = i(k')$, $S[s] = q$, else $S[s]$ is
the string obtained by removing the first symbol of $P[s]$ (which must be $x$).

\item if $x \cdot q = \epsilon$ and $t, \ell = 0$, then $x \in \{\epsilon, \epsilon^{\sf w}\}$.

\end{enumerate}
\end{enumerate}
\end{definition}

\begin{remark}
Every execution of the interval-concurrent queue with no dequeue
operation returning~$\epsilon^{\sf w}$ is an execution of the
sequential queue.
\end{remark}

The proof of the following lemma is
similar to the proof of Lemmas~\ref{lemma-props-set-seq-stack}
and~\ref{lemma-props-set-seq-queue}.

\begin{lemma}
\label{lemma-props-int-seq-queue}
Let $A$ be any interval-linearizable implementation of the
interval-concurrent queue with weak-empty.  Then,
\begin{enumerate}

\item All sequential executions of $A$ are executions of the sequential queue.

\item All executions in which no \Deq operation is concurrent with any
  other operation are linearizable with respect to the sequential
  queue.
\end{enumerate}
\end{lemma}

The algorithm in Figure~\ref{figure:int-seq-queue}, which we call
\IntSeqQueue, is an interval-linearizable wait-free implementation of
a queue with weak-emptiness, which uses base objects with consensus
number two.  \IntSeqQueue is a simple modification of \SeqQueue in which
an \Enq operation proceeds as in \SeqQueue, while a \Deq operation
scans  $Items$ at most two times to obtain an item, in both cases
recording the number of taken items. If the two numbers are the same
(cf. \emph{double clean} scan), then the operations return
$\epsilon$, otherwise
it returns $\epsilon^{\sf w}$.

\begin{figure}[ht]
\centering{ \fbox{
\begin{minipage}[t]{150mm}
\scriptsize
\renewcommand{\baselinestretch}{2.5} \resetline
\begin{tabbing}
aaaaa\=aaa\=aaa\=aaa\=aaa\=aaa\=aaa\=\kill 

{\bf Shared Variables:}\\

\> $Tail:$ {\sf Fetch\&Inc} base object initialized to 1\\

\> $Items[1, \hdots ]:$ array of \SWAP base objects initialized to $\bot$\\ \\

{\bf Operation}  $\Enq(x_i)$ {\bf is} \\

\line{PP01} \> $tail_i \leftarrow Tail.\FAI()$\\

\line{PP02} \> $Items[tail_i].\W(x_i)$\\

\line{PP03} \> {\bf {\sf return} } {\sf true}\\

{\bf end} \Enq\\ \\

{\bf Operation}  $\Deq()$ {\bf is} \\

\line{PP05} \> {\bf for} $k \leftarrow 1$ {\bf up to} $2$ {\bf do}\\

\line{PP06} \> \> $taken_i[k] \leftarrow 0$\\

\line{PP07} \> \> $tail_i \leftarrow Tail.\R()-1$\\

\line{PP08} \> \> {\bf for} $r_i \leftarrow 1$ {\bf up to} $tail_i$ {\bf do}\\

\line{PP09}  \> \> \> $x_i \leftarrow Items[r_i].\R()$\\

\line{PP10}  \> \> \> {\bf if} $x_i \neq \bot$ {\bf then}\\

\line{PP11}  \> \> \> \> $x_i \leftarrow Items[r_i].\SWAP(\top)$\\

\line{PP12}  \> \> \> \> {\bf if} $x_i \neq \top$ {\bf
     then {\sf return}} $x_i$ {\bf end if}\\

\line{PP13}  \> \> \> \> $taken_i[k] \leftarrow taken_i[k] + 1$\\

\line{PP14}  \> \> \> {\bf end if}\\

\line{PP15} \> \> {\bf end for}\\

\line{PP16} \> {\bf end for}\\

\line{PP17} \> {\bf if} $taken_i[1] = taken_i[2]$ {\bf then} {\bf {\sf return}} $\epsilon$\\

\line{PP18} \> \> {\bf else} {\bf {\sf return}} $\epsilon^{\sf w}$\\

\line{PP19} \>  {\bf end if}\\

{\bf end} \Deq

\end{tabbing}
\end{minipage}
}
\caption{\small Wait-free interval-concurrent queue from consensus
  number 2 (code for $p_i$).}
\label{figure:int-seq-queue}
}
\end{figure}

\begin{theorem}
\label{theo-int-seq-queue}
The algorithm \IntSeqQueue (Figure~{\em{\ref{figure:int-seq-queue}}})
is a wait-free interval-linearizable
implementation of the queue with weak-empty, using objects with
consensus number two.
\end{theorem}

\begin{proof}
It is clear that \IntSeqQueue is wait-free as all its base objects are
wait-free, thus we focus on showing it is interval-linearizable.  Let
$E$ be any execution of \IntSeqQueue. Since the algorithm is
wait-free, there is an extension of $E$ in which all its pending
operations are completed, and no new operation is started. Observe
that any interval-linearization of such an  extension is an
interval-linearization of $E$. Thus, without loss of generality, we
can assume all operations in $E$ are completed.

Consider the sequence $F$ obtained by removing from $E$ every \Deq
operation (invocation, response and steps) that returns $\epsilon^{\sf w}$.  
Observe that none of these operations changes the state of
$Tail$ and $Items$ (as they return $\epsilon^{\sf w}$), hence $F$ is
indeed an execution of \IntSeqQueue. Furthermore, $F$ is an execution
of \SeqQueue: \Enq operations behave the same in \SeqQueue and
\IntSeqQueue, and every \Deq operation scans $Items$ at most twice 
and either returns an item or finds the queue empty and hence returns
$\epsilon$. 
Since \SeqQueue is linearizable, consider a linearization $\Lin(F)$ of it.

We will obtain an interval-linearization $\IntLin(E)$ of $E$ from
$\Lin(F)$ by adding to it the \Deq operations of $E$ that return
$\epsilon^{\sf w}$.  The idea is the following. When any such
operation, say $\sf DEQ^{\epsilon^{w}}$, starts, the queue is in a
state $q$.  Since $\sf DEQ^{\epsilon^{w}}$ returns $\epsilon^{\sf w}$,
it is the case that $\sf DEQ^{\epsilon^{w}}$ runs concurrently to \Deq
operations that take the items in $q$, and other \Deq operations that
take the new items that are concurrently enqueued while $\sf
DEQ^{\epsilon^{w}}$ runs.  Roughly speaking, if $q \neq \epsilon$,
$\sf DEQ^{\epsilon^{w}}$ will be interval-linearized with all the \Deq
operations that take the items in $q$; otherwise, $\sf
DEQ^{\epsilon^{w}}$ will be interval-linearized alone, since the queue
is empty in this case.

To obtain $\IntLin(E)$, below we suppose, without loss of generality,
that $\Lin(F)$ has the following property, which, roughly speaking,
say that \Enq operations are linearized in $\Lin(F)$ ``as late as
possible''.

\begin{assumption}
\label{asmp}
If $\Lin(F)$ has two consecutive operations $\langle \Enq(x):{\sf
  true} \rangle$ and $\langle \Deq():y \rangle$ that are concurrent in
$F$ with $x \neq y$, then the order of the operations can be exchanged
to obtain another sequential execution of the queue which is a
linearization of $F$ too.  Thus, we assume that $\Lin(F)$ does not
have such a pair of operations.
\end{assumption}

Let $\sf DEQ^{\epsilon^{w}}$ be any \Deq operation of $E$ that returns
$\epsilon^{\sf w}$.  Let $\alpha$ be the shortest prefix of $\Lin(F)$
with every operation $\sf op$ of $F$ with ${\sf op} <_E {\sf
  DEQ^{\epsilon^{w}}}$; in words, $\alpha$ is the first moment in
which all operations that happen before $\sf DEQ^{\epsilon^{w}}$ in
$E$ are linearized in $\Lin(F)$.

In what follows, let $\sf OP^{last}$ denote the operation at the end
of $\alpha$ (if there is one).  First, we observe that ${\sf
  OP^{last}} <_E {\sf DEQ^{\epsilon^{w}}}$: if not, the prefix of
$\Lin(F)$ obtained by removing $\sf op$ from $\alpha$ has all
operations in $F$ that happen before ${\sf DEQ^{\epsilon^{w}}}$ in
$E$, which contradicts minimality of the length of $\alpha$. Then:

\begin{claim}
\label{claim-case1-1}
If $\alpha \neq \epsilon$, ${\sf OP^{last}} <_E {\sf DEQ^{\epsilon^{w}}}$.
\end{claim}

We also note that for every ${\sf op} \neq {\sf OP^{last}}$ in
$\alpha$ (if there is one), either ${\sf op} <_E {\sf
  DEQ^{\epsilon^{w}}}$ or ${\sf op} \, ||_E \, {\sf
  DEQ^{\epsilon^{w}}}$: if not, then ${\sf DEQ^{\epsilon^{w}}} <_E
{\sf op}$ and then ${\sf OP^{last}} <_E {\sf op}$, which leads to a
contradiction as $\sf op$ appears before $\sf OP^{last}$ in $\Lin(F)$
and thus the partial order $<_E$ is not respected in $\Lin(F)$,
implying that $\Lin(F)$ is not a linearizability of~$F$.  Thus, we
have:

\begin{claim}
\label{claim-case1-2}
If $\alpha \neq \epsilon$, for every ${\sf op} \neq {\sf OP^{last}}$
in $\alpha$, either ${\sf op} <_E {\sf DEQ^{\epsilon^{w}}}$ or ${\sf
  op} \, ||_E \, {\sf DEQ^{\epsilon^{w}}}$.
\end{claim}

Let $q$ be the state of the queue at the end of $\alpha$. 
We split the rest of proof in the following two cases:

\paragraph{Case $q = \epsilon$.} 
This case is simple. Recall that if the queue is empty, a \Deq
operation can be linearized at a single point and
non-deterministically return either $\epsilon$ or $\epsilon^{\sf w}$
(case 2.e in the definition of the interval-concurrent queue with
weak-empty).  Thus, ${\sf DEQ^{\epsilon^{w}}}$ is linearized in
$\IntLin(E)$ alone between the last operation in $\alpha$ and the next
one in $\Lin(F)$. Since $q = \epsilon$ at the end of $\alpha$, ${\sf
  DEQ^{\epsilon^{w}}}$ returning~$\epsilon^{\sf w}$ follows the
specification of the interval-concurrent queue with weak-empty.
Moreover, the linearization of ${\sf DEQ^{\epsilon^{w}}}$ in
$\IntLin(E)$ respects the partial order $<_E$, by
Claim~\ref{claim-case1-2} and since $\alpha$ contains every operations
that happens before $F$.  If there is another ${\sf
  DEQ'^{\epsilon^{w}}}$ having the same prefix $\alpha$, then ${\sf
  DEQ^{\epsilon^{w}}}$ and ${\sf DEQ'^{\epsilon^{w}}}$ are linearized
respecting $<_E$, namely, if they are concurrent, their order does not
matter, otherwise, we follows the order imposed by~$<_E$.

\paragraph{Case $q = x_1 x_2 \cdots x_t \neq \epsilon$.} 
This is the hard case.
Since $q \neq \epsilon$, $\alpha$ must contain the \Enq operations enqueuing 
the items in $q$, and thus $\alpha \neq \epsilon$.

\begin{claim}
\label{claim-case1-3}
For every item $x_i$ in $q$, there is a \Deq operation in $F$ that
dequeues $x_i$.
\end{claim}

We observe first that it is enough to prove that there is a \Deq
operation that dequeues~$x_t$.  As a corollary, we get that the same
happens for every item $x_i$ in $q$, namely, there is a \Deq operation
in $F$ that dequeues $x_i$: $\Lin(F)$ is a sequential execution of the
queue and thus $x_t$ can be dequeued only if $x_i$ (which is enqueued
before $x_t$ as it appears before $x_t$ in $q$) is dequeued first.

Consider the operation $\langle \Enq(x_t):{\sf true} \rangle$ in $\alpha$.
We identify the following subcases:
\begin{itemize}
\item ${\sf OP^{last}} = \langle \Enq(x_t):{\sf true} \rangle$ or
  $\langle \Enq(x_t):{\sf true} \rangle <_E {\sf DEQ^{\epsilon^{w}}}$.
  In any case, we have $\langle \Enq(x_t):{\sf true} \rangle <_E {\sf
    DEQ^{\epsilon^{w}}}$ and thus $x_t$ is stored in some entry in
  $Items$ before ${\sf DEQ^{\epsilon^{w}}}$ starts scanning $Items$,
  and thus there must exists a \Deq operations that takes $x_t$ before
  ${\sf DEQ^{\epsilon^{w}}}$ reads the entry where $x_t$ is written.

\item ${\sf OP^{last}} \neq \langle \Enq(x_t):{\sf true} \rangle$ and
  ${\sf DEQ^{\epsilon^{w}}} <_E \langle \Enq(x_t):{\sf true} \rangle$.
  We already saw that ${\sf OP^{last}} <_E {\sf DEQ^{\epsilon^{w}}}$,
  from which follows that ${\sf OP^{last}} <_E \langle \Enq(x_t):{\sf
    true} \rangle$. Since $\langle \Enq(x_t):{\sf true} \rangle$
  appears before ${\sf OP^{last}}$ in $\Lin(F)$, we have that the
  partial order $<_E$ is not respected in $\Lin(F)$, which is a
  contradiction as $\Lin(F)$ is a linearizability of $F$.

\item ${\sf OP^{last}} \neq \langle \Enq(x_t):{\sf true} \rangle$ and
  $\langle \Enq(x_t):{\sf true} \rangle \, ||_E \, {\sf
  DEQ^{\epsilon^{w}}}$ (see Figure~\ref{fig-ex-case-non-empty-state}). 
  Let $Items[pos]$ be the entry where $x_t$ is
  written by $\langle \Enq(x_t):{\sf true} \rangle$.  By
  contradiction, suppose that there is no \Deq operation in $F$ that
  dequeues $x_t$.  Observe that the only way this can happen is that
  $\langle \Enq(x_t):{\sf true} \rangle$ writes $x_t$
  (Line~\ref{PP02}) after ${\sf DEQ^{\epsilon^{w}}}$ has scanned two
  times that entry (i.e. it executes Line~\ref{PP09} with $Items[pos]$
  in the two iterations of the for loop);
  Figure~\ref{fig-ex-case-non-empty-state} schematizes this (and the
  rest of the argument for this case).

Observe the following: 
\begin{itemize}
\item all operations of $\alpha$ after
  $\langle \Enq(x_t):{\sf true} \rangle$ are \Deq operations 
(if not, then $x_t$ would not be the last item in $q$), and
\item each of these \Deq operations dequeues an item that is not in $q$
  (if not, that item would not appear in $q$).
\end{itemize}

\begin{figure}[ht]
\begin{center}
\includegraphics[width=12.5cm]{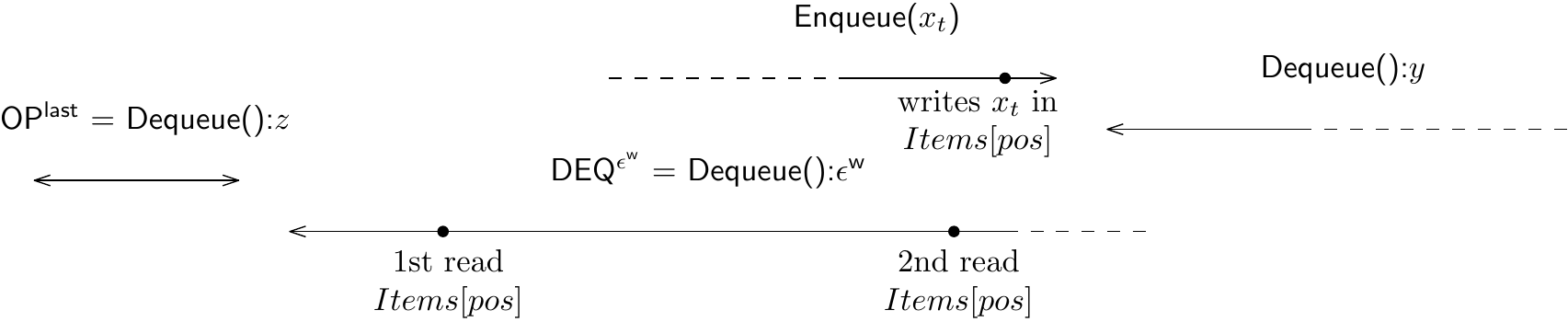}

\caption{\small Example of case
  ${\sf OP^{last}} \neq \langle \Enq(x_t):{\sf true} \rangle$ and
  $\langle \Enq(x_t):{\sf true} \, ||_E \, {\sf DEQ^{\epsilon^{w}}}$.}
\label{fig-ex-case-non-empty-state}
\end{center}
\end{figure}

Thus, we have ${\sf OP^{last}} = \langle \Deq():z \rangle$.  Let
$\langle \Deq():y \rangle$ be the operation of $\alpha$ right after
$\langle \Enq(x_t):{\sf true} \rangle$.  We argue that ${\sf
  OP^{last}} = \langle \Deq():z \rangle <_E \langle \Deq():y \rangle$,
which leads to a contradiction as $\langle \Deq():y \rangle$ appears
before $\langle \Deq():z \rangle$ in $\Lin(F)$ and thus the partial
order $<_E$ is not respected in $\Lin(F)$, implying that $\Lin(F)$ is
not a linearizability of $F$.

We already saw that $\langle \Deq():z \rangle <_E {\sf
  DEQ^{\epsilon^{w}}}$, and thus the response of $\langle \Deq():z
\rangle$ in $F$ appears before the first read of ${\sf
  DEQ^{\epsilon^{w}}}$ to entry $Items[pos]$ (where $x_t$ is written).
We also observe that $\langle \Enq(x_t):{\sf true} \rangle$ and
$\langle \Deq():y \rangle$ are not concurrent in $F$, by
Assumption~\ref{asmp}; thus it must be that $\langle \Enq(x_t):{\sf
  true} \rangle <_E \langle \Deq():y \rangle$, from which follows that
the invocation of $\langle \Deq():y \rangle$ appears in $F$ after the
second read of ${\sf DEQ^{\epsilon^{w}}}$ to entry
$Items[pos]$. Therefore, ${\sf OP^{last}} = \langle \Deq():z \rangle
<_E \langle \Deq():y \rangle$ (see
Figure~\ref{fig-ex-case-non-empty-state}), which concludes the
argument for this subcase, and completes the proof of
Claim~\ref{claim-case1-3}.
\end{itemize}

We now define the interval-linearization of $\sf DEQ^{\epsilon^{w}}$
in $\IntLin(E)$ (see Figure~\ref{fig-ex-int-lin-a}).  Let $\beta$ be
the sequence of $\Lin(F)$ containing all its operations from the
operation right after $\sf OP^{last}$ to $\langle \Deq():x_t \rangle$
(whose existence is shown above).  The interval-linearization of $\sf
DEQ^{\epsilon^{w}}$ spans the interval of $\Lin(F)$ defined by
$\beta$; formally, the invocation of $\sf DEQ^{\epsilon^{w}}$ is added
to the first concurrency class of $\beta$ and the response of it to
the last concurrency class of $\beta$.  By construction, at the
beginning of the interval-linearization of $\sf DEQ^{\epsilon^{w}}$,
the state of the queue is $q$ and at the end of it, all items in $q$
are dequeued, which follows the specification of the
interval-concurrent queue with weak-empty.

To conclude this case, we show that for every operation $\sf op$ of
$\beta$, ${\sf op} \, ||_E \, {\sf DEQ^{\epsilon^{w}}}$ in
Claim~\ref{claim-case1-4}, which together with
Claim~\ref{claim-case1-2} implies that the interval-linearization of
$\sf DEQ^{\epsilon^{w}}$ respects the partial order $<_E$.

\vspace{0.5cm}
\begin{figure}[ht]
\begin{center}
\includegraphics[width=10.5cm]{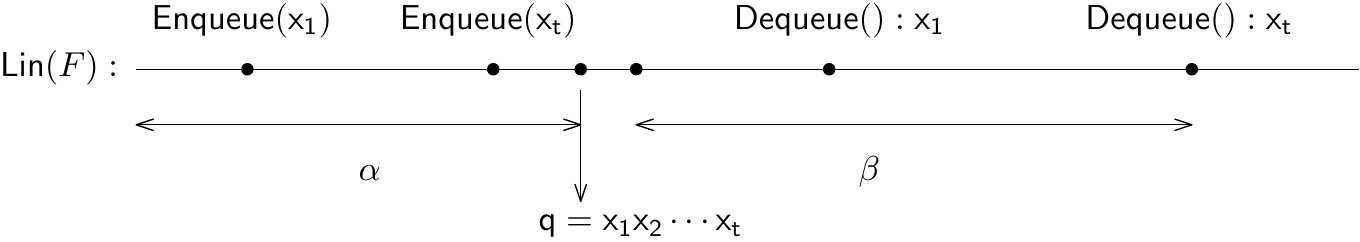}
\caption{\small Interval-linearization of
  $\sf DEQ^{\epsilon^{w}}$ when $q \neq \epsilon$.}
\label{fig-ex-int-lin-a}
\end{center}
\end{figure}

\begin{claim}
\label{claim-case1-4}
For every operation $\sf op$ of $\beta$,
${\sf op} \, ||_E \, {\sf DEQ^{\epsilon^{w}}}$.
\end{claim}

There are two subcases to be proven:
\begin{itemize}
\item It is not the case that ${\sf op} <_E {\sf
  DEQ^{\epsilon^{w}}}$. If so, then $\alpha$ is not the shortest
  prefix of $\Lin(F)$ containing all operations that happen before
  $\sf DEQ^{\epsilon^{w}}$ in $F$, which is a contradiction.

\item It is not the case that ${\sf DEQ^{\epsilon^{w}}} <_E {\sf
  op}$. For contradiction, suppose that ${\sf DEQ^{\epsilon^{w}}} <_E
  {\sf op}$.  Let $Item[pos]$ the entry where $x_t$ is written by
  $\langle \Enq(x_t):{\sf true} \rangle$, and consider the operation
  $\langle \Deq():x_t \rangle$.

\begin{itemize}

\item First, suppose that $\langle \Enq(x_t):{\sf true} \rangle$
  writes $x_t$ in $Items[pos]$ (hence, $Items[pos]$ is 
 $Items[tail_i]$ if the invoking process is $p_i$, 
  Line~\ref{PP02}) before ${\sf   DEQ^{\epsilon^{w}}}$
  performs its second read of $Items[pos]$
  (Line~\ref{PP08} corresponding to the second iteration of the for
  loop).  Note that ${\sf DEQ^{\epsilon^{w}}}$ might or might not
  obtain $x_t$ from that read operation, however, certainly it must
  happens that $\langle \Deq():x_t \rangle$ takes $x_t$ (i.e. it
  successfully obtains $x_t$ from $Items[pos]$ in Line~\ref{PP10})
  before ${\sf DEQ^{\epsilon^{w}}}$ completes. Thus, in $F$ the
  response of $\langle \Deq():x_t \rangle$ appears before the response
  of ${\sf DEQ^{\epsilon^{w}}}$. Observe that either $\langle
  \Deq():x_t \rangle <_E {\sf DEQ^{\epsilon^{w}}}$ or $\langle
  \Deq():x_t \rangle \, ||_E \, {\sf DEQ^{\epsilon^{w}}}$. Thus ${\sf
    op} \neq \langle \Deq():x_t \rangle$ as ${\sf DEQ^{\epsilon^{w}}}
  <_E {\sf op}$. Moreover, we have that $\langle \Deq():x_t \rangle
  <_E {\sf op}$.  This is a contradiction as $\sf op$ appears before
  $\langle \Deq():x_t \rangle$, which contradicts that $\Lin(F)$
  respects~$<_E$.

\item Otherwise, note that either $\langle \Enq(x_t):{\sf true}
  \rangle \, ||_E \, {\sf DEQ^{\epsilon^{w}}}$ or ${\sf
    DEQ^{\epsilon^{w}}} <_E \langle \Enq(x_t):{\sf true} \rangle$.
  Thus, $\sf OP^{last} \neq \langle \Enq(x_t):{\sf true} \rangle$
  because $\sf OP^{last} <_E {\sf DEQ^{\epsilon^{w}}}$, by
  Claim~\ref{claim-case1-1}.  We first observe that it cannot happen
  ${\sf DEQ^{\epsilon^{w}}} <_E \langle \Enq(x_t):{\sf true} \rangle$,
  because if so, $\sf OP^{last} <_E \langle \Enq(x_t):{\sf true}
  \rangle$ but in $\Lin(F)$ $\langle \Enq(x_t):{\sf true} \rangle$
  appears first and then $\sf OP^{last}$, which contradicts that
  $\Lin(F)$ is a linearization of $F$.  Thus, we have $\langle
  \Enq(x_t):{\sf true} \rangle \, ||_E \, {\sf DEQ^{\epsilon^{w}}}$.
  As observed above (third subcase in the proof of
  Claim~\ref{claim-case1-3}), all operations of $\alpha$ after
  $\langle \Enq(x_t):{\sf true} \rangle$ are \Deq operations, and each
  of these operations dequeues an item that is not in $q$.  Let $\sf
  OP^{last} = \langle \Deq():z \rangle$ and $\langle \Deq():y \rangle$
  be the operation in $\alpha$ right after $\langle \Enq(x_t):{\sf
    true} \rangle$.  By Assumption~\ref{asmp}, we must have $\langle
  \Enq(x_t):{\sf true} \rangle <_E \langle \Deq():y \rangle$.  We also
  have $\langle \Deq():z \rangle <_E {\sf DEQ^{\epsilon^{w}}}$, by
  Claim~\ref{claim-case1-1}.  Therefore, we obtain that $\langle
  \Deq():z \rangle <_E \langle \Deq():y \rangle$, which leads to a
  contradiction because $\langle \Deq():y \rangle$ appears before
  $\langle \Deq():z \rangle$ in $\Lin(F)$, implying that $\Lin(F)$ is
  not a linearization of $F$.
\end{itemize}
The second subcase follows and hence Claim~\ref{claim-case1-4} follows too.
\end{itemize}

To obtain $\IntLin(E)$, we repeat the construction above for any such
operation $\sf DEQ^{\epsilon^{w}}$ of $E$.  To conclude the proof of
the Theorem, we prove that $\IntLin(E)$ respects the partial order
$<_E$.  Since the order of operations in $\Lin(F)$ is not modified to
built $\IntLin(E)$, we only need to check that any two \Deq operations
returning $\epsilon^{\sf w}$ respect the partial order $<_E$.  Namely,
given any two ${\sf DEQ_1^{\epsilon^{w}}} <_E {\sf
  DEQ_2^{\epsilon^{w}}}$ such operations, it holds that the
interval-linearizations of the operations in $\IntLin(E)$ do not
overlap and the interval-linearization of ${\sf DEQ_1^{\epsilon^{w}}}$
appears before the interval-linearization of ${\sf
  DEQ_2^{\epsilon^{w}}}$.

Let $\alpha_i$ be the prefix of $\Lin(F)$ used to define the
interval-linearizability of ${\sf DEQ_i^{\epsilon^{w}}}$.  Recall that
$\alpha_i$ is the shortest prefix of $\Lin(F)$ containing all
operations in $F$ that happens before ${\sf DEQ_i^{\epsilon^{w}}}$ in
$E$.  First note that $\alpha_2$ cannot be a proper prefix of
$\alpha_1$ because ${\sf DEQ_1^{\epsilon^{w}}} <_E {\sf
  DEQ_2^{\epsilon^{w}}}$.

Consider first the case where the interval-linearization of
${\sf DEQ_1^{\epsilon^{w}}}$ boils down to a single linearization  point.
This can happen only if the state
of the queue at the end of $\alpha_1$ is empty. Thus, if $\alpha_1 =
\alpha_2$, the interval-linearization of ${\sf DEQ_2^{\epsilon^{w}}}$
is a point too and, by construction, ${\sf DEQ_2^{\epsilon^{w}}}$ is
linearized after ${\sf DEQ_1^{\epsilon^{w}}}$. And if $\alpha_1 \neq
\alpha_2$, then the interval-linearization of ${\sf
  DEQ_2^{\epsilon^{w}}}$, which might be a point or an interval,
necessarily appears before the linearization of ${\sf
  DEQ_1^{\epsilon^{w}}}$, by construction.

Consider now the case that the interval-linearization of
${\sf  DEQ_1^{\epsilon^{w}}}$ is an interval (a sequence of linearization points). 
 Then, the state of the queue
at the end of $\alpha_1$ is $q_1 = x_1 x_2 \cdots x_t \neq \epsilon$.
The interval-linearization of ${\sf DEQ_1^{\epsilon^{w}}}$ is the
interval $\beta_1$ of $\Lin(F)$ starting at the next operation after
$\alpha_1$ and ending at the \Deq operation that returns the item
$x_t$. Observe that there is no moment in $\beta_1$ in which the queue
is empty, and thus, if the interval-linearization of ${\sf
  DEQ_1^{\epsilon^{w}}}$ is a point, it cannot belong to $\beta_1$.

The only case that remains to be analyzed is when both
interval-linearizations of ${\sf DEQ_1^{\epsilon^{w}}}$ and ${\sf
  DEQ_2^{\epsilon^{w}}}$ are intervals. Let $\sf OP^{last}$ be the
last operation $\alpha_1$ and consider $\langle \Enq(x_t):{\sf true}
\rangle$ that appears in $\alpha_1$ too. For the sake of
contradiction, suppose that the interval-linearizations overlap.
Thus, $\langle \Deq():x_t \rangle$ belong to both
intervals. This implies that in $E$, $\langle \Deq():x_t \rangle$
takes $x_t$ (i.e. it successfully executes Lines~\ref{PP11}
and~\ref{PP12}) when ${\sf DEQ_2^{\epsilon^{w}}}$ is running.  Thus,
$\langle \Deq():x_t \rangle$ starts before of concurrently with ${\sf
  DEQ_1^{\epsilon^{w}}}$ and ends concurrently with ${\sf
  DEQ_2^{\epsilon^{w}}}$.  We analyze what happens with $\langle
\Enq(x_t):{\sf true} \rangle$:

\begin{itemize}

\item $\langle \Enq(x_t):{\sf true} \rangle = {\sf OP^{last}}$.
  This case cannot happen because if so,
  $\langle \Enq(x_t):{\sf true} \rangle <_E {\sf DEQ_1^{\epsilon^{w}}}$,
  by Claim~\ref{claim-case1-1}, and when ${\sf DEQ_1^{\epsilon^{w}}}$
  starts, $x_t$ would be already in $Items$ and ${\sf DEQ_1^{\epsilon^{w}}}$
  would take it, as $\langle \Deq():x_t \rangle$ 
takes its item only after ${\sf DEQ_2^{\epsilon^{w}}}$ starts.

\item $\langle \Enq(x_t):{\sf true} \rangle \neq {\sf OP^{last}}$. In
  this case, consider the operation of $\langle \Deq():y \rangle$
  right after $\langle \Enq(x_t):{\sf true} \rangle$ (it was already
  observed above that operation must be a \Deq with $y$ not appearing
  in $q_1$).  By Assumption~\ref{asmp}, $\langle \Enq(x_t):{\sf true}
  \rangle <_E \langle \Deq():y \rangle$. If $\langle \Deq():y \rangle
  = {\sf OP^{last}}$, then we have $\langle \Enq(x_t):{\sf true}
  \rangle <_E {\sf DEQ_1^{\epsilon^{w}}}$, because $\langle \Deq():y
  \rangle <_E {\sf DEQ_1^{\epsilon^{w}}}$, by
  Claim~\ref{claim-case1-1}. As observed in the previous case, this
  cannot happen.

We only remain to consider that $\langle \Deq():y \rangle \neq {\sf
  OP^{last}}$. Since $\langle \Enq(x_t):{\sf true} \rangle$ appears
before ${\sf OP^{last}}$ in $\Lin(F)$, we must have either $\langle
\Enq(x_t):{\sf true} \rangle <_E {\sf OP^{last}}$ or $\langle
\Enq(x_t):{\sf true} \rangle \, ||_E \, {\sf OP^{last}}$.  If $\langle
\Enq(x_t):{\sf true} \rangle <_E {\sf OP^{last}}$, then $\langle
\Enq(x_t):{\sf true} \rangle <_E {\sf DEQ_1^{\epsilon^{w}}}$, and we
already saw that leads to a contradiction; thus, that cannot happen.
Finally, if $\langle \Enq(x_t):{\sf true} \rangle \, ||_E \, {\sf
  OP^{last}}$, then the first step of $\langle \Enq(x_t):{\sf true}
\rangle$ (Line~\ref{PP02}) occurs in $E$ before ${\sf OP^{last}}$
terminates but its last step (Line~\ref{PP03}) can occur only after
${\sf DEQ_1^{\epsilon^{w}}}$ terminates because, as already explained,
$\langle \Enq(x_t):{\sf true} \rangle$ takes $x_t$ after ${\sf
  DEQ_2^{\epsilon^{w}}}$ starts.  Since $\langle \Enq(x_t):{\sf true}
\rangle <_E \langle \Deq():y \rangle$, by Assumption~\ref{asmp}, we
thus have ${\sf OP^{last}} <_E \langle \Deq():y \rangle$. We have
reached a contradiction because $\langle \Deq():y \rangle$ appears
before ${\sf OP^{last}}$, and then $\Lin(F)$ does not respect the
partial order $<_F$.
\end{itemize}

We conclude that \IntSeqQueue is interval-linearizable and wait-free.
This concludes the proof of the theorem.
\end{proof}

Observe that the definition of the interval-concurrent queue
with weak-empty allows a \Deq operation to be linearized at a single
point and non-deterministically return either $\epsilon$ or
$\epsilon^{\sf w}$, when the queue is empty (case 2.e in Definition~\ref{def-weak-empty}).  
The proof of Theorem~\ref{theo-int-seq-queue} indeed interval-linearizes
some \Deq operations returning $\epsilon^{\sf w}$ at  single points,
although the operations return that value due to concurrency.  
Certainly, we can remove case 2.e in Definition~\ref{def-weak-empty} and modify the proof of
Theorem~\ref{theo-int-seq-queue} so that every interval-linearization
of a \Deq operation returning $\epsilon^{\sf w}$ is an interval, at
the cost of making the proof more complex and longer.

\subsection{Interval-Concurrent Queue with Weak-emptiness and Multiplicity}

Using the techniques in Sections~\ref{sec-set-stack}
and~\ref{sec-set-queue}, we can obtain a \R/\W wait-free
implementation of a even more relaxed interval-concurrent queue in
which an item can be taken by several dequeue operations, i.e., with
multiplicity. In more detail, the interval-concurrent queue with
weak-emptiness is modified such that concurrent \Deq operations can
return the same item and are set-linearized in the same concurrency
class, as in Definitions~\ref{def-set-stack} and~\ref{def-set-queue}.

We obtain a \R/\W wait-free interval-concurrent implementation of the
queue with weak-emptiness and multiplicity by doing the following:
(1) replace the \FAI object in \IntSeqQueue with a \R/\W wait-free \COUNT,
(2) extend $Items$ to a matrix to handle collisions, and 
(3) simulate the \SWAP operation with a \R followed by a \W.
Thus, we have:

\begin{theorem}
There is a \R/\W wait-free interval-linearizable implementation of the
queue with weak-emptiness  and multiplicity.
\end{theorem}

\section{Final Discussion}
\label{sec-final-discussion}
Considering classical data structures initially defined for sequential
computing, this work has introduced new well-defined relaxations 
to adapt them to concurrency and  investigated 
algorithms that implement them on top of "as weak as possible" base operations.
It has first introduced the notion of set-concurrent queues and stacks with 
multiplicity, a relaxed version of queues and tasks in which an item
can be dequeued more than once by concurrent operations.  Non-blocking
and wait-free set-linearizable implementations were presented,
both based only on
the simplest \R/\W operations.  These are the first
implementations of relaxed queues and stacks using only these
operations.  The implementations imply algorithms for idempotent
work-stealing and out-of-order stacks and queues.  

The paper also
introduced a relaxed concurrent queue with weak-emptiness
check, which allows a dequeue operation to return a ``weak-empty certificate''
reporting that the queue might be empty.  A wait-free
interval-linearizable implementation using
objects with consensus number two was presented for such a relaxed queue. 
As there are only
non-blocking linearizable (not relaxed)
queue implementations using objects with consensus
number two,  it is an open question if there is such a wait-free implementation. 
The proposed queue relaxation allowed us to go from
non-blocking to wait-freedom using only objects with consensus number two.

Considering the consensus numbers 1 and 2 of the base objects on top
of  which are built the \AC{relaxations} of the stack and the queue presented in the present article,
Table~\ref{table:stack} and Table~\ref{table:queue} give a global view of
the tradeoff between
the  liveness and safety properties of the constructed stacks and  queues. 

\newcommand{\push}{\sf{push}}
\newcommand{\pop}{\sf{pop}}
\newcommand{\enqueue}{{\sf{enqueue}}}
\newcommand{\dequeue}{{\sf{dequeue}}}

\begin{table}[h!]
\begin{center}
 \begin{tabular}{|c|c|c|c|}
\hline
Base object  &  Liveness      & Safety  &  Algorithm  \\

\hline
\hline
CN = 1  &  $\push()$: wait-freedom  &  $\push()$: linearizability
& Figures \ref{figure:set-seq-stack} and \ref{figure:set-seq-stack-improve}\\
   &  $\pop()$: wait-freedom   &  $\pop()$: set-linearizability &\\
\hline
CN =  2  &  $\push()$: wait-freedom  &  $\push()$: linearizability
                                           & Figure  \ref{figure:stack}  \cite{AGM07}
\\
&     $\pop()$: wait-freedom   &  $\pop()$: linearizability &\\
\hline
\end{tabular}
\end{center}
\caption{Stack in the consensus number (CN) 1 and 2 worlds.}
\label{table:stack}
\end{table}

\begin{table}[h!]
\begin{center}
 \begin{tabular}{|c|c|c|c|}
\hline
Base object  &  Liveness      & Safety  &  Algorithm  \\ 

\hline
CN = 1  &  $\enqueue()$: wait-freedom  &  $\enqueue()$: linearizability
                                  & Figure~\ref{figure:set-seq-queue} \\
&  $\dequeue()$: non-blocking    &  $\dequeue()$:  set-linearizability  &  \\
\hline
\AC{CN = 1}  &  $\enqueue()$: wait-freedom  &  $\enqueue()$: linearizability
                                  & Derived from Figure~\ref{figure:int-seq-queue} \\
&  $\dequeue()$: wait-freedom    &  $\dequeue()$:  interval-linearizability  &  \\
\hline
CN = 2  &  $\enqueue()$: wait-freedom  &  $\enqueue()$: linearizability
                                            & Figure~\ref{figure:queue}  (modified \\

&     $\dequeue()$: non-blocking   &  $\dequeue()$: linearizability 
                                    &  version of \cite{L01}) \\
\hline
CN = 2  &  $\enqueue()$: wait-freedom  &  $\enqueue()$: linearizability
                                    & Figure~\ref{figure:int-seq-queue} \\
   &  $\dequeue()$: wait-freedom  &  $\dequeue()$: interval-linearizability &\\
\hline
\end{tabular}
\end{center}
\caption{Queue in the consensus number (CN) 1 and 2 worlds.}
\label{table:queue}
\end{table}

As far as we know, there are only non-blocking queue
implementations using objects with consensus number two, and it is an
open question if there is such a wait-free implementation. The
proposed queue relaxation allowed us to go from non-blocking to
wait-free.

This work also can be seen as prolonging the results described
in~\cite{CRR18} where the notion of interval-linearizability was
introduced and  set-linearizability~\cite{N94} was investigated.  
It has shown that linearizability, set-linearizability and
interval-linearizability constitute a hierarchy of
consistency conditions that allow us to formally express the behavior
of non-trivial (and still meaningfull)
relaxed queues and stacks on top of simple base
objects such as \R/\W registers.

An interesting extension to this work is to explore if the proposed
relaxations can lead to practical efficient implementations. Another
interesting extension is to explore if set-concurrent or
interval-concurrent relaxations of other concurrent data structures
would allow implementations to be designed without requiring the stronger
computational power
provided by  atomic {\sf Read-Modify-Write} operations.

\bibliographystyle{plainurl}


\appendix

\section{Correctness proof of \SeqQueue  (Fig.~\ref{figure:set-seq-stack},
  Section~\ref{LF-lin-queue-CN=2})}
\label{app-proof-seq-queue}

As already explained, \SeqQueue is a slight modification of Li's
non-blocking queue implementation~\cite{L01}.  The only difference
between the two implementations is that \SeqQueue relaxes the
condition for returning $\epsilon$: while Li's implementations
requires $tail_i = tail'_i \wedge taken_i = taken'_i$ in the condition
in Line~\ref{P16}, \SeqQueue only requires $taken_i = taken'_i$.
Roughly speaking, we show below that any execution of
\SeqQueue can be modified so that it corresponds to an execution of
Li's queue implementation, and hence any linearization of it is also a
linearization of the execution of \SeqQueue.

Let $E$ be any execution of \SeqQueue. Since \SeqQueue is non-blocking,
there is an extension of $E$ in which all operations are completed and
no new operations started.  Thus, we can assume that all operations in
$E$ are completed.  For every \Deq operation in $E$, let $tail'_i$ be
the value of $tail_i$ in the second-to-last while iteration, or zero
if there is only one iteration.  Let $k$ be the number of \Deq
operations returning $\epsilon$ with $tail'_i$ being distinct to
$tail_i$.  By induction on $k$, we show that $E$ is linearizable.

For the base case, $k=0$, for every \Deq operation returning
$\epsilon$, it holds $tail_i = tail'_i$ (additionally to $taken_i =
taken'_i$), namely, it satisfies the condition in Li's queue, and thus
$E$ is indeed an execution of that algorithm, which implies that $E$ 
is linearizable.  Assuming  the claim holds when $E$ has $k-1$ such 
operations, let us  show it holds for $k$.  We will modify $E$ such that it
has $k-1$ such operations. Figure~\ref{fig-example-3} exemplifies the
transformation.

\begin{figure}[ht]
\begin{center}
\includegraphics[scale=0.7]{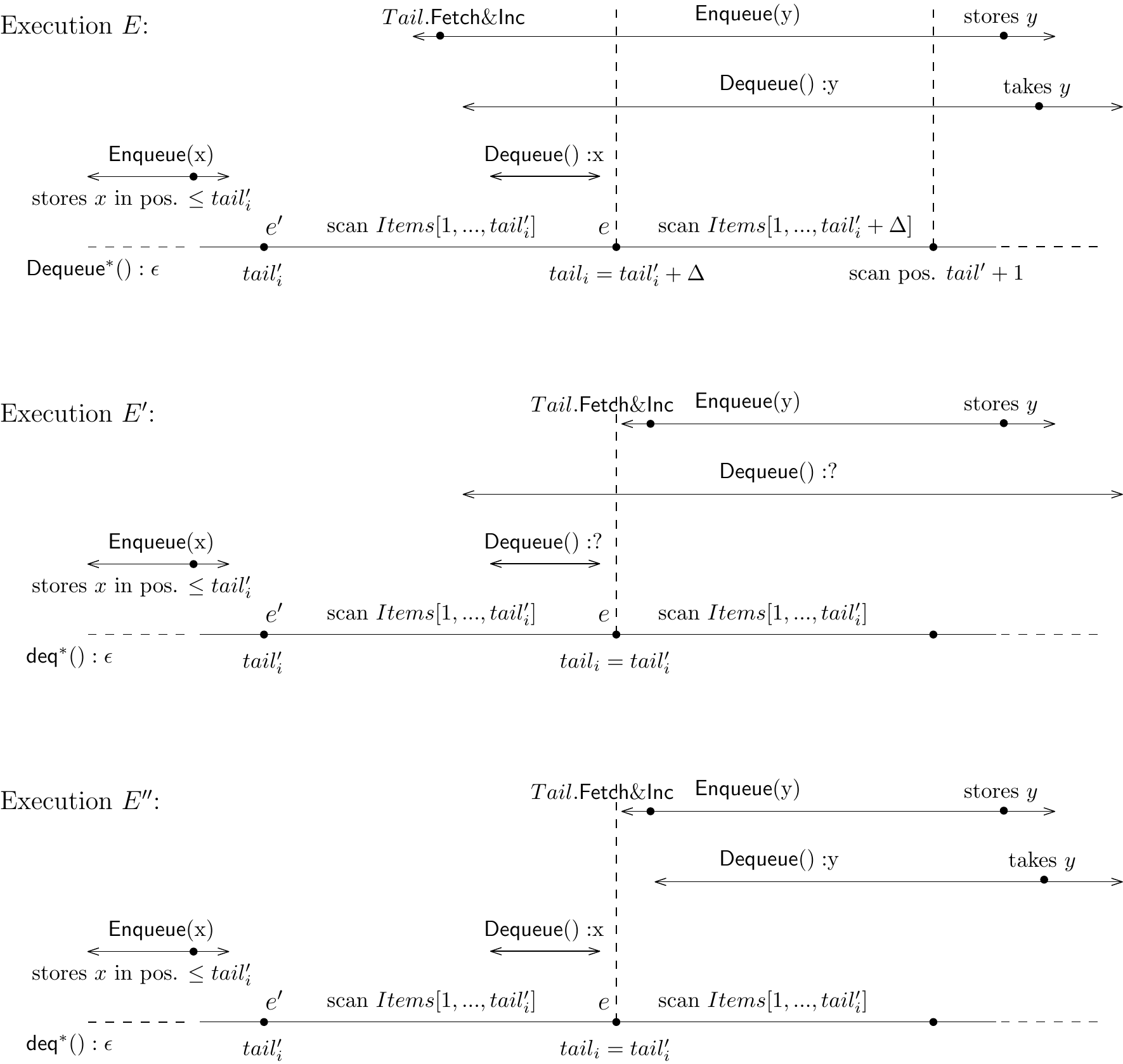}
\caption{\small An example of the transformation from \SeqQueue to
  Li's queue.}
\label{fig-example-3}
\end{center}
\end{figure}

Among those $k$ \Deq operations in $E$ (all of them returning
$\epsilon$), let $\sf deq^*$ denote the one that executes first its
last step in Line~\ref{P07}; let $e$ denote such a step and let $e'$
denote the previous step of $\sf deq^*$ corresponding to
Line~\ref{P07} (or to Line~\ref{P04} if $\sf deq^*$ executes only one
iteration of the while loop).  We have $taken_i = taken'_i$ and
$tail_i \neq tail'_i$ when $\sf deq^*$ returns $\epsilon$ in
Line~\ref{P16}.  Thus, $tail_i = tail'_i + \Delta$, for some integer
$\Delta \geq 1$, which can only happen if there is at least one \Enq
operation that executes its step in Line~\ref{P01} in the interval $I$
of $E$ from step $e'$ to step~$e$.  Let $\sf enq^*$ be any of such
\Enq operations.  Note that $\sf enq^*$ writes its item,
Line~\ref{P02}, after $\sf deq^*$ has scanned all entries in $Items[1,
  \hdots, tail'_i]$ in its last while iteration (if not, $\sf deq^*$
finds $taken_i > taken'_i$ and it does not decides $\epsilon$ in that
iteration).  For any such operation $\sf enq^*$, we move forward its
step in Line~\ref{P01}, right after step $e$ of $\sf deq^*$,
respecting the relative order among all the steps.  Thus, in the
resulting execution $E'$, the step $e$ of $\sf deq^*$ reads the same
value from $Tail$, name, it finds $tail_i = tail'_i$ (additionally to
$taken_i = taken'_i$), hence it holds Li's queue condition for $\sf
deq^*$ and returns $\epsilon$ in $E'$ too.  Therefore, $E'$ has $k-1$
\Deq operations that return $\epsilon$ with $tail'_i$ being distinct
to $tail_i$.  Observe that operations in $E$ and $E'$ have the same
real-time order.  By induction hypothesis, $E'$ is linearizable,
however, a linearization of $E'$ might not be a linearization of $E$
because it could be that, after the modifications, a \Deq operation
returns distinct values in $E$ and $E'$.  Thus, to conclude, we modify
$E'$ to be sure that all \Deq operations in $E'$ return the same
value.

Let $f$ be a step of a \Deq operation (distinct from $\sf deq^*$)
corresponding to Line~\ref{P07} and appearing in the interval $I$ in
$E$.  If $f$ reads the value $tail'_i$, then that step obtains the
same value in $E'$ and there is nothing to do; intuitively, $f$
happens before all \Enq operations that execute Line~\ref{P01} in $I$,
and hence all modifications to $E$ happen after $f$.  If $f$ reads a
value greater than $tail'_i$, then there are two sub-cases (see
Figure~\ref{fig-example-3}).
If the \Deq operation returns in the while iteration $f$ belongs to,
an item that is stored in an entry of $Items[1, \hdots, tail'_i]$,
then we do not move $f$, because in $E'$ that step reads the value
$tail'_i$, which is fine since it just needs to scan up to that entry
in $Items$ to obtain its item.  Otherwise, we move forward $f$ right
after step $e$ of $\sf deq^*$ (and all its scan steps that appear in
$I$), respecting the relative order among all the steps in $E$.  Thus,
$f$ reads the same value in $E$ and in the resulting execution. Note
that we can do this because $\sf deq^*$ finds $taken_i = taken'_i$
when it decides, and hence the state of $Items[1, \hdots, tail_i]$
does not change between $e$ and the first read step (Line~\ref{P09})
of the last scan of $\sf deq^*$.  We do the same for each such \Deq
operation. Let $E''$ be the resulting execution.  Therefore, the
operations $E$ and $E''$ have the same real-time order and return the
same values.  By induction hypothesis, $E''$ is linearizable, and any
linearization of $E''$ is a linearization of $E$ too.  The claim
follows.

\end{document}